\pdfoutput=1
\RequirePackage{ifpdf}
\ifpdf % We are running pdfTeX in pdf mode
\documentclass[pdftex]{sigma}
\else
\documentclass{sigma}
\fi

\numberwithin{equation}{section}

\newtheorem{Theorem}{Theorem}[section]
\newtheorem{Corollary}[Theorem]{Corollary}
\newtheorem{Lemma}[Theorem]{Lemma}
\newtheorem{Proposition}[Theorem]{Proposition}
 { \theoremstyle{definition}
\newtheorem{Definition}[Theorem]{Definition}
\newtheorem{Remark}[Theorem]{Remark} }

\usepackage{mathtools}

\usepackage{enumitem}
\usepackage{float}
\usepackage{tikz-cd}
\usepackage[all]{xy}
\usepackage{pb-diagram,pb-xy}

\usetikzlibrary{arrows}
\tikzset{
commutative diagrams/.cd,
arrow style=tikz,
diagrams={>=latex}}

\newcommand{\g}{\mathfrak{g}}

\newcommand{\longto}{\longrightarrow}
\newcommand{\cL}{{\cal L}}
\newcommand{\cT}{\mathcal{T}}

\newcommand{\R}{\mathbb{R}}
\newcommand{\C}{\mathbb{C}}
\newcommand{\kf}{\mathfrak{k}}
\newcommand{\pf}{\mathfrak{p}}
\newcommand{\uf}{\mathfrak{u}}

\newcommand{\Lin}{{\rm Lin}}
\newcommand*\circled[1]{\tikz[baseline=(char.base)]{
 \node[shape=circle,draw,inner sep=1pt] (char) {#1};}}

\definecolor{kellygreen}{rgb}{0.3, 0.73, 0.09}

\begin{document}
\allowdisplaybreaks

\newcommand{\arXivNumber}{1909.10088}

\renewcommand{\PaperNumber}{046}

\FirstPageHeading

\ShortArticleName{Routh Reduction of Palatini Gravity in Vacuum}

\ArticleName{Routh Reduction of Palatini Gravity in Vacuum}

\Author{Santiago CAPRIOTTI}

\AuthorNameForHeading{S.~Capriotti}

\Address{Departamento de Matem\'atica, Instituto de Matem\'atica de Bah\'ia Blanca (INMABB),\\ CONICET, Universidad Nacional del Sur, Av.~Alem 1253, 8000 Bah\'ia Blanca, Argentina}
\Email{\href{mailto:santiago.capriotti@uns.edu.ar}{santiago.capriotti@uns.edu.ar}}

\ArticleDates{Received September 24, 2019, in final form May 11, 2020; Published online May 30, 2020}

\Abstract{An interpretation of Einstein--Hilbert gravity equations as Lagrangian reduction of Palatini gravity is made. The main technique involved in this task consists in representing the equations of motion as a set of differential forms on a suitable bundle. In this setting Einstein--Hilbert gravity can be considered as a kind of Routh reduction of the underlying field theory for Palatini gravity. As a byproduct of this approach, a novel set of conditions for the existence of a vielbein for a given metric is found.}


\Classification{53C80; 53C05; 83C05; 70S05; 70S10}

\section{Introduction}\label{sec:Intro}

The relationship between Einstein--Hilbert and Palatini formulation of gravity has been studied by several authors. It could be established in the Lagrangian formulation by comparing the action functionals for each case~-- see for example \cite{Romano1993, tsamparlis_palatini_1978} and references therein. From the Hamiltonian point of view, the main theoretical tool used in the discussion of the connection between these formulations of gravity appears to be some flavor of Hamiltonian reduction. For instance, \cite{2012GReGr..44.2337D} and~\cite{Romano1993} use ADM formalism \cite{citeulike:820116} in order to establish the connection; it has also been explored in \cite{Cattaneo2019,Ibort:2016xoo}, where the correspondence is set by using a Hamiltonian structure on the set of fields at the boundary.

From this viewpoint, it becomes interesting to find a reduction scheme relating the Lagrangian formulation of Palatini and Einstein--Hilbert gravity directly, without the detour through Hamiltonian formalism. So far, there exist two ways to implement reduction at the Lagrangian level, namely \emph{Lagrange--Poincar\'e reduction} \cite{CASTRILLONLOPEZ2013352, 2003CMaPh.236..223L, lopez00:_reduc_princ_fiber_bundl, ellis11:_lagran_poinc} and \emph{Routh reduction} \cite{Capriotti201723, Capriotti2019, CrampinMestdagRouth,GarciaToranoAndres2016291,eprints21388, Marsden00reductiontheory}. Moreover, there are physical considerations that can be said in support of this kind of reduction: They deal with not only the reduction, but also the reconstruction problem, and it is argued in~\cite{Nawarajan:2016vzv} that reconstruction can be relevant from the physical point of view.

Routh reduction was originally designed as a technique to get rid of cyclic variables in a~Lagrangian function by using a fixed value of the corresponding momentum~\cite{leopoldalexanderpars1965}. The mechanical system thus obtained has a configuration space where the cyclic variables are absent and whose Lagrangian has a term forcing the momentum to have the prescribed value; the new Lagrangian function is called \emph{Routh function} or \emph{Routhian}. This reduction scheme can be further generalized to symmetries characterized by a non abelian Lie group~-- see~\cite{Marsden00reductiontheory} and references therein. In the general case, the equations of motion of the system obtained by quotient out the symmetry are modified by the addition of a new term, called \emph{force term}, which is related to the curvature of a~principal bundle associated to the symmetry. It should be noted that the Routhian reduces to the Lagrangian and the force term annihilates when the momentum value is set to zero.

These versions of Routh reduction require some sort of regularity to the Lagrangian in order to construct the Routhian; this regularity is used to translate the constraints into the momentum variables to constraints into the velocities. Therefore, it is possible to avoid this regularity condition by working in a unified setting, where variables associated to both velocities and momenta are taken into account; this approach can be used to set a Routh reduction procedure not suffering this limitation \cite{GarciaToranoAndres2016291}. This viewpoint proves its usefulness in~\cite{Capriotti2019}, where unified formalism was employed in the generalization of Routh reduction to the field theory realm. In fact, in order to improve the understanding of the tools and strategies used in the present work, it could be helpful to give an account of both the problem addressed in this reference and the techniques used to solve it:
\begin{itemize}\itemsep=0pt
\item \textbf{Result in \cite{Capriotti2019}:} For the moment, we will represent a variational problem with a
 pair $(\pi\colon E\to M,\cL)$, where $\pi\colon E\to M$ is
 a bundle and $\cL\colon J^1\pi\to\wedge^m(T^*M)$ is a
 Lagrangian density; in this setting, we can define the action
\[
 \Gamma\pi\ni s\mapsto\int_U\cL\circ j^1s
\]
for any compact $U\subset M$, and the variational problem becomes
\[
 \delta\int_U\cL\circ j^1s=0
\]
for arbitrary variations $\delta s$. Let us suppose that there exists a Lie group acting freely and vertically on $E$, such that the Lagrangian density is invariant by this action, yielding to a momentum map $J\colon J^1\pi\to\wedge^{m-1}\big(T^*J^1\pi\big)\otimes\g^*$. In order to construct the reduced Lagrangian theory, it is necessary to fix a value for the momentum map through the specification of a $\g^*$-valued $(m-1)$-form
\[
 \mu\in\Omega^{m-1} (M,\g^* ),
\]
and also to choose a connection $\omega$ on the principal bundle $p_G^E\colon E\to\Sigma:=E/G$. Let us indicate by $G_\mu\subset G$ the stabilizer of the momentum $\mu$. Then the main result in~\cite{Capriotti2019} established a correspondence between the extremals of the variational problem $(\pi\colon E\to M,\cL)$ and
\[
 \big(\overline{\pi}\colon E/G_\mu\times_M\mathop{{\rm Lin}}{(\overline{\pi}^*TM,\widetilde{\g})}\to M,\mathcal{R}_\mu\big).
\]
This variational problem is the reduced variational problem for Routh reduction in field theory. As we warned before, when dealing with variational problems obtained through Routh reduction, it could be necessary to add a \emph{force term}, in order to take care of the non trivial nature of the bundles involved~\cite{kharlamov_characteristic_1977}. In short, the local Routhian functions constructed through Routh original procedure cannot be pasted into a global Routhian function unless the bundle obtained by the symmetry quotient is trivial. It is thus necessary to use a connection in order to define this global function; as a by-product, a force term associated to the curvature of the chosen connection should be taken into account. In any case, the force term is in this context represented by a $(m+1)$-form $\phi_\mu$ on $J^1\overline{\pi}$ constructed from $\omega$ and the momentum~$\mu$. Therefore, the underlying variational problem becomes modified by an external force
\[
 \delta\int_U\mathcal{R}_\mu\circ j^1\overline{s}=\int_U\langle \phi_\mu,\delta\overline{s}\rangle.
\]
\item \textbf{Techniques:} As was said above, it is useful to add momentum variables in order to avoid regularity issues when dealing with Routh reduction. It leads us to work with the unified formulation of field theory \cite{2004JMP....45..360E}, where the space of velocities~$J^1\pi$ is enlarged with the space of multimomentum~-- see equation~\eqref{eq:NotationVerticalForms} below for an explanation of the symbols involved in this formula:
 \[
 J^1\pi^\circ:=\wedge^{m}_2(E)/\wedge^m_1(E)
 \]
 to form a bundle
 \begin{gather}\label{eq:WClassic}
 W_{\text{cl}}:=J^1\pi\times_M J^1\pi^\circ.
 \end{gather}
 In this space, a variational problem can be set. A word of caution should be given here: The variational problem on $W_{\text{cl}}$ is different to those described by a pair $(\pi\colon E\to M,\cL)$ above. Concretely, in this new setting we have a Lagrangian $\Theta_\cL$ that is not a density, but a $m$-form on~$W_{\text{cl}}$; the associated action becomes
 \[
 \Gamma (W_{\text{cl}}\to M)\ni\sigma\mapsto\int_U\sigma^*\Theta_{\cL}.
 \]
 A crucial difference with the variational problems of Lagrangian field theory arises from this action: When performing variations, they should be free because in $W_{\text{cl}}$ there are no variables that can be identified as velocities. It claims for a definition of variational problem capable to associate geometrical structures to these differences; this is captured by the notion of \emph{Griffiths variational problem} considered in \cite{book:852048,hsu92:_calcul_variat_griff} and will be crucial in our approach to reduction~-- see Definition~\ref{def:GriffithsVarProb} below.

Beyond better regularity properties, the unified setting has several advantages stemming from the linear structures that come equipped with it, and the formulation of relevant data in terms of differential forms, on which pullback maps allow for straightforward comparisons of equations and other quantities.
\end{itemize}

Now, we are looking for a Routh reduction scheme similar to the one described above, in order to relate the variational problems for Palatini and Einstein--Hilbert gravity. This forces us to discuss how we will represent the flavors of gravity to work with. For Palatini gravity we will use the formulation given in \cite{capriotti14:_differ_palat}, where it was interpreted as an example of the concept of Griffiths variational problem. In order to describe this formulation, let us summarize the main characteristics of a Griffiths variational problem. It consists into three kind of data: A bundle $p\colon W\to M$, whose sections will be the fields of the theory, a Lagrangian form $\lambda\in\Omega^m(W)$ setting the dynamics, and a set of forms $\mathcal{I}\subset\Omega^\bullet(W)$ (more precisely, an \emph{exterior differential system}) describing the set of differential restrictions on the fields. This additional geometric structure replaces the contact structure of the jet bundle; therefore, the usual constraints
\[
 u^A_\mu=\frac{\partial u^A}{\partial x^\mu}
\]
are absent. It should be stressed that this change has profound consequences on the equations governing the extremals of the underlying variational problem: In general, when performing variations, it results that
\[
 \delta u^A_\mu\not=\big(\delta u^A\big)_\mu.
\]
In fact, because the set of sections verifying the differential restrictions represented by~$\mathcal{I}$ must be invariant by variations, this relationship should be replaced by conditions describing the infinitesimal symmetries of this set.

Accordingly, a Griffiths variational problem will be described with the symbol
\begin{gather}\label{rem:griffiths-notation}
(p\colon W\to M,\lambda,\mathcal{I}),
\end{gather}
where the additional data $\mathcal{I}$ has been added. The variational problem underlying such triple consists in finding the extremals of the action
\[
 S [s ]:=\int_Ms^*\lambda,
\]
where the sections $s\colon M\to W$ of the bundle~$p$ must be integral for the set of forms in $\mathcal{I}$, namely,
\[
 s^*\alpha=0
\]
for every $\alpha\in\mathcal{I}$.

We can use this approach to represent Palatini gravity as a first order field theory. Recall that for this type of gravity theory, the degrees of freedom are a vielbein and a connection; it suggests taking the bundle of frames $\tau\colon LM\to M$ (whose sections are the \emph{vielbein}) as the field bundle. Moreover, its jet space~$J^1\tau$ can be decomposed as
\[
 J^1\tau=LM\times_M C (LM ),
\]
where $C(LM)\to M$ is the connection bundle, that is, a bundle whose sections are the principal connections of~$LM$ -- see Section~\ref{sec:LFT} for details. It means that sections of $\tau_1\colon J^1\tau\to M$ are composed by a pair vielbein~$+$ connection and so, they are suitable for the description of the fundamental degrees of freedom in Palatini gravity. Now, in Palatini formulation, the vielbein and the connection are not independent; they are related by two constraints: on the one side, the connection must be torsionless; on the other, the metric associated to a given vielbein should be invariant respect to the parallel translation associated to the given connection. So, instead of working with the total jet space~$J^1\tau$, we will work in a submanifold $\cT_0\subset J^1\tau$, namely, the submanifold corresponding to the \emph{torsion zero constraint}; also, the contact structure is changed by a set of differential constraints implementing the \emph{metricity conditions}~-- see equation~\eqref{eq:MetricityConds} for a geometrical definition of these conditions, and Section~\ref{sec:zero-tors-subm} for a discussion of its physical meaning. It should be stressed that the zero torsion condition in this approach becomes part of the bundle of fields, rather than appearing as part of the equations of motion, as is usually the case (see for example \cite{hehl_general_1976, Peldan:1993hi}).

The underlying variational problem is given by the triple
\[
 (\tau_1'\colon \cT_0\to M,\mathcal{L}_{\rm PG},\langle \omega_\pf\rangle ),
\]
where $\mathcal{L}_{\rm PG}$ is a Lagrangian $m$-form on~$\cT_0$ and $\langle \omega_\pf\rangle$ is an exterior differential system on~$\cT_0$ encoding the metricity conditions.

It should be stressed that, when working with such variational problem, the allowed variations are not those fullfilling the commutativity property
\[
 \delta u^A_\mu=\big(\delta u^A\big)_\mu;
\]
instead, they must be thought as infinitesimal symmetries of the metricity conditions. We will not explore further this topic here; the interested reader can find a detailed study of these questions in~\cite{Capriotti:2019bqh}. On the other hand, the usual variational problem for Einstein--Hilbert gravity is given by the triple
\[
 \big( (\tau_\Sigma )_2\colon J^2\tau_\Sigma\to M,\cL_{\rm EH},\mathcal{I}_{{\rm con}}^\Sigma\big)
\]
where $\tau_\Sigma\colon \Sigma\to M$ is the bundle of metrics with a given signature, $\cL_{\rm EH}$ is the (second order) Lagrangian density~\cite{doi:10.1063/1.4998526} and~$\mathcal{I}_{{\rm con}}^\Sigma$ is the contact structure on the second order bundle~$J^2\tau_\Sigma$. Although it would be possible in principle to work with a second order jet, it is easier to deal with first order jets; so, we are forced to find a first order variational problem describing Einstein--Hilbert gravity. In fact, one of the results of the present article is to prove that there exists an $m$-form (not a density) $\cL_{\rm EH}^{(1)}$ on $J^1\tau_\Sigma$ such that
\[
 \big((\tau_\Sigma)_1\colon J^1\tau_\Sigma\to M,\cL_{\rm EH}^{(1)},\mathcal{I}_{{\rm con}}^\Sigma\big)
\]
also describes Einstein--Hilbert gravity~-- see Section~\ref{sec:first-order-vari}, where a correspondence between the extremals for these problems will be established.

In summary, a large part of the present article consists in repeating for Palatini gravity what was done for first order field theory in~\cite{Capriotti2019}; more concretely and in line with the previous discussion, we want to answer an specific problem through a particular set of techniques:
\begin{itemize}\itemsep=0pt
\item \textbf{Problem:} To establish correspondences between the extremals of the variational problem
 \[
 \big(\tau_1'\colon \cT_0\to M,\mathcal{L}_{\rm PG},\langle \omega_\pf\rangle\big),
 \]
 describing Palatini gravity, with the extremals of the variational problem
 \[
 \big((\tau_\Sigma)_2\colon J^2\tau_\Sigma\to M,\cL_{\rm EH},\mathcal{I}_{{\rm con}}^\Sigma\big)
 \]
 for Einstein--Hilbert gravity.
\item \textbf{Techniques:} The following diagram could help us to outline the strategy involved in proving the correspondences:
 \begin{equation}\label{eq:TechniquesDiagram}
 \begin{tikzcd}[row sep=1.3cm,column sep=3.5cm]
 (W_{\rm PG},\lambda_{\rm PG},0)
 \arrow[leftrightarrow,dashed]{r}{}
 \arrow[leftrightarrow,swap]{dd}{\circled{1}}
 &
(W_{\rm EH},\lambda_{\rm EH},0)
 \arrow[leftrightarrow]{d}{\circled{2}}
 \\
 &
 \big(J^1\tau_\Sigma,\cL_{\rm EH}^{(1)},\mathcal{I}_{{\rm con}}^\Sigma\big)
 \arrow[leftrightarrow]{d}{\circled{3}}
 \arrow[shift left=1ex,dotted]{ld}{\text{Reconstruction}}
 \\
 (\cT_0,\mathcal{L}_{\rm PG}, \langle\omega_{\pf} \rangle )
 \arrow[shift left=1ex,dotted]{ru}{\text{Reduction}}
 &
 \big(J^2\tau_\Sigma,\cL_{\rm EH},\mathcal{I}_{{\rm con}}^\Sigma\big)
 \end{tikzcd}
 \end{equation}
 The main idea is to replace the Griffiths variational problems representing Palatini and Einstein--Hilbert gravity with other kind of variational problems playing the role, in this generalized context, of the unified formalism. This operation is indicated schematically by arrows~$\circled{1}$ and $\circled{2}$. As we described above, formulation of the unified version of field theory requires to pass to a new phase space where multivelocities and multimomentum variables are represented. Nevertheless, the product bundle construction~\eqref{eq:WClassic} is not versatile enough to establish the correspondences~between the variational problems mentioned above because it depends on the interpretation of the bundle~$J^1\pi^\circ$ as a dual of the jet space $J^1\pi$. Instead, these correspondences could be studied by using a very general setting described in~\cite{GotayCartan}, called \emph{Lepage-equivalent variational problem}. It consists in lifting the variational problem $\big(\pi_1\colon J^1\pi\to M,\cL,\mathcal{I}_{{\rm con}}\big)$ to a variational problem on an affine subbundle of the bundle of $m$-forms on~$J^1\pi$, and this subbundle becomes isomorphic to (a~quotient of) $W_{\text{cl}}$. Concretely, let us fix an adapted coordinate chart $\big(\phi=\big(x^\mu,u^A\big),U\big)$ on~$E$; then we have a set of $(m-1)$-forms defined as follows
 \[
 \eta:={\rm d}x^1\wedge\cdots\wedge {\rm d}x^m,\qquad\eta_\nu:=\frac{\partial}{\partial x^\nu}\lrcorner\eta.
 \]
 In this setting, the Lepage-equivalent problem of Gotay is given by the data
 \[
(\pi_0\colon W_0\to M,\lambda_0)
 \]
 where
 \begin{gather}\label{eq:LepageEquivProblemGeneral}
 W_0:=\cL+{I}^m_{{\rm con}}\subset\wedge^m\big(T^*J^1\pi\big)
 \end{gather}
 and, locally, the bundle of forms $I^m_{{\rm con}}$ is given by
 \[ I^m_{{\rm con}}\big|_{(x^\mu,u^A,u^A_\nu)}:=\big\{p_A^\mu\big({\rm d}u^A-u^A_\nu {\rm d}x^\nu\big)\wedge\eta_\mu\colon p_A^\mu\in\mathbb{R}\big\}\subset\wedge^m\big(T^*_{(x^\mu,u^A,u^A_\nu)}\big(J^1\pi\big)\big).
 \]
 The Lagrangian form $\lambda_0$ arises from the fact that $W_0$ is a subbundle of the bundle of $m$-forms $\wedge^m\big(T^*J^1\pi\big)$; it has a canonical $m$-form $\lambda$, and $\lambda_0$ is its pullback to $W_0$. The underlying action becomes
 \[
 \Gamma\pi_0\ni\sigma\mapsto\int_U\sigma^*\lambda_0
 \]
 and the variations must be performed in an entirely free way; that is, in this case, the variations $\delta u^A$ and $\delta u^A_\mu$ on $J^1\pi$ are independent. As we said above, the coordinates along the fibers of the map
 \[
 q\colon \ W_0\to J^1\pi
 \]
 correspond to multimomentum variables, and we can see that
 \[
 W_0\simeq J^1\pi\times_MJ^1\pi^\circ,
 \]
 namely, the bundle $W_0$ is isomorphic to the total bundle~$W_{\text{cl}}$ in the unified formulation of first order field theory. Thus, it can be readily seen that Lepage-equivalent variational problems serve as a generalization of the unified formalism to variational problems of Griffiths type; the interested reader is referred to~\cite{GotayCartan} for details.

Part of the techniques we will use to establish the relationship (reduction and reconstruction) between the extremals of these variational problems are analogous to those that were employed in~\cite{Capriotti2019}: To transform the variational problems into its Lepage-equivalent problems, represented by the two-way arrows $\circled{1}$ and $\circled{2}$ in diagram \eqref{eq:TechniquesDiagram}, and to use the operations available there in order to set the correspondences. As we indicated above, in order to apply this construction to Einstein--Hilbert gravity, it will be necessary to consider an equivalent first-order variational problem instead; this operation is shown as arrow~$\circled{3}$ in the previous diagram. The rest of the techniques involved in the proof of reduction/reconstruction are represented by the dashed arrow: The basic idea is to take advantage of the fact that the variational problems we are trying to connect are formulated on an affine subbundles of a bundle of forms. As we said above, the naturality of forms respect to the pullback operation is crucial to compare the equations governing the extremals in every variational problem. Although in our present article this characteristic will not be used, it is still true that the vector bundle structure enhances this comparison by allowing the translation along a given form.
\end{itemize}

These considerations set the purposes of the following article: On the one hand, to carry out a proof of concept for the generalization of Routh reduction to variational problems more general than those corresponding to first order field theory, generalizing the techniques employed in~\cite{Capriotti2019}; on the other hand, to apply Routh reduction of field theory in the context of a meaningful example, namely, a formulation of gravity with basis. In short, the results achieved in the article could help to put the relationship between Einstein--Hilbert and Palatini gravity in a precise geometrical framework, which can in principle be extended to more sophisticated field theoretic analyses.

In this vein, it should be mentioned that among the objectives in mind while planning the article there was one of mathematical nature; in short, to contribute to the improvement of the language in a geometrical sense around the general question on how Einstein--Hilbert and Palatini versions of gravity are related. In this regard, it appears that the subject has continuously been developed since these formulations were found, but the use of a more precise geometrical language in describing the new achievements in the area has grown at a slower pace~-- see \cite{doi:10.1063/1.4890555,doi:10.1063/1.4998526, sardanashvily_classical_2011,vey_multisymplectic_2015} for examples dealing with the geometrical perspective in gravity. To some extent, this imbalance could place some obstacles to the natural development of some aspects of the subject; in these cases, the geometrical viewpoint might be more useful than a mere calculation in their understanding. Hence, the present article intends to contribute to a better comprehension of the correspondence between Einstein--Hilbert and Palatini gravity, through the formulation of this correspondence in a geometrical framework.

The rest of the paper is organized as follows: In Section~\ref{sec:vari-probl-unif} fundamentals of Griffiths variational problems and their relationship with the unified formalism in field theory are given. In Section~\ref{sec:LFT} geometrical tools necessary for the construction of the variational problem for Palatini gravity we will use in this article are reviewed. In Section~\ref{sec:vari-probl-palat} the actual construction of this variational problem, as well as the associated unified problem, is done. The symmetry considerations necessary to carry out the reduction are discussed in Section \ref{sec:symmetry-momentum}. Section~\ref{sec:local-coord-expr} is rather technical, containing some calculations used in the reduction and reconstruction theorems. In Section~\ref{sec:metr-cont-struct} the results achieved in the previous section are employed in the search of identifications between geometrical structures present in both the reduced and unreduced spaces: A remarkable fact in this vein is that the metricity constraints correspond with the contact structure of a jet bundle after projection onto the quotient. Construction of the first order formalism for Einstein--Hilbert gravity (and its correspondence with the usual second order formalism) is delayed until Section~\ref{sec:first-order-vari}; also, a unified formalism for this variational problem is discussed in this section. The choice of a connection induces a splitting in the contact structure on the jet space of the frame bundle. In Section~\ref{sec:cont-bundle-decomp} the effects of this splitting in the variational formulation of Palatini gravity are analyzed. In Section~\ref{sec:routhian} the Routhian is constructed, showing that the Routhian for Palatini gravity is the (first order) Einstein--Hilbert Lagrangian. Finally, in Section~\ref{sec:einst-hilb-grav} the reduction theorem and the reconstruction theorem are proved. The main result of this section is the notion of \emph{flat condition for a metric}, which is a helpful hypothesis in the proof of the reconstruction theorem. Also, it is proved there that this condition is equivalent to the parallelizability of the spacetime manifold $M$; therefore, the reconstruction scheme can be carried out only locally.

\subsubsection*{Notations} We are adopting here the notational conventions from \cite{saunders89:_geomet_jet_bundl} when dealing with bundles and its associated jet spaces. Also, if $Q$ is a manifold, $\Lambda^p (Q)=\wedge^p(T^*Q)$ denotes the $p$-th exterior power of the cotangent bundle of $Q$. Moreover, for $k\leq l$ the set of $k$-horizontal $l$-forms on the bundle $\pi\colon P\to N$ is
\begin{gather}\label{eq:NotationVerticalForms}
 \wedge^l_k(P):=\big\{\alpha\in\wedge^l(P)\colon v_1\lrcorner\cdots v_k\lrcorner\alpha=0\text{ for any }v_1,\dots,v_k\text{ }\pi\text{-vertical vectors}\big\}.
\end{gather}
For the same bundle, the set of vectors tangent to $P$ in the kernel of $T\pi$ will be represented with the symbol $V\pi\subset TP$. In this regard, the set of vector fields which are vertical for a bundle map $\pi\colon P\to N$ will be indicated by $\mathfrak{X}^{V\pi}(P)$. The space of differential $p$-forms, sections of $\Lambda^p (Q)\to Q$, will be denoted by $\Omega^p(Q)$. {We also write $\Lambda^\bullet(Q)=\bigoplus_{j=1}^{\dim Q}\Lambda^j(Q)$}. If $f\colon P\to Q$ is a smooth map and $\alpha_x$ is a $p$-covector on $Q$, we will sometimes use the notation $\alpha_{f(x)}\circ T_xf$ to denote its pullback $f^*\alpha_x$. If $P_1\to Q$ and $P_2\to Q$ are fiber bundles over the same base $Q$ we will write $P_1\times_Q P_2$ for their fibred product, or simply $P_1\times P_2$ if there is no risk of confusion. Unless explicitly stated, the canonical projections onto its factor will be indicated by
\[
 {\rm pr}_i\colon \ P_1\times P_2\to P_i,\qquad i=1,2
\]
or, in general
\[
 {\rm pr}_{i_1\cdots i_k}\colon \ P_1\times\cdots\times P_L\to P_{i_1}\times\cdots\times P_{i_k}
\]
for any collection of different indices $i_1,\dots,i_k\in\{1,\dots,L\}$.

Given a manifold $N$ and a Lie group $G$ acting on $N$, the symbol $[n]_G$ for $n\in N$ will indicate the $G$-orbit in $N$ containing $n$; the canonical projection onto its quotient will be denoted by
\[
 p_G^N\colon \ N\to N/G.
\]
Also, if $\mathfrak{g}$ is the Lie algebra for the group~$G$, the symbol $\xi_N$ will represent the infinitesimal generator for the $G$-action associated to~$\xi\in\mathfrak{g}$. Finally, Einstein summation convention will be used everywhere.

\section{Variational problems and unified formalism}\label{sec:vari-probl-unif}

As we mentioned in Section~\ref{sec:Intro}, the scheme used for Routh reduction relies on the notion of unified formulation of a variational problem. So, we will devote the present section to describe the construction of a unified formalism for a particular family of variational problem, the so called \emph{Griffiths variational problems}.
\begin{Definition}[Griffiths variational problem]\label{def:GriffithsVarProb}
A \emph{Griffiths variational problem} is a triple
\[
(p\colon W\to M,\lambda,\mathcal{I})
\]
where $p\colon W\to M$ is a bundle, $\lambda\in\Omega^m(W)$ is a form on~$W$ (here~$m$ is the dimension of the base manifold $M$) and $\mathcal{I}$ is an exterior differential system. The underlying variational problem consists in finding the extremals of the action
\[
 \Gamma p\ni\sigma\mapsto\int_U\sigma^*\lambda,
\]
where $U\subset M$ is any compact submanifold and $\sigma$ is a section integral for $\mathcal{I}$, namely,
\[
 \sigma^*\alpha=0
\]
for any $\alpha\in\mathcal{I}$.
\end{Definition}

\begin{Remark}A particular instance of a Griffiths variational problem is provided by the so called \emph{classical variational problem}, which is the variational problem underlying the first order classical field theory~\cite{Gotay:1997eg}: In it, the underlying bundle is the jet space $\pi_1\colon J^1\pi\to M$ associated to a bundle $(E,\pi,M)$, the Lagrangian form is induced on $\Omega^m\big(J^1\pi\big)$ by the Lagrangian density
\[
 \cL\colon \ J^1\pi\to\wedge^m (T^*M)
\]
and the set of differential restrictions is imposed by the so called \emph{contact structure on the jet bundle}. Particularly, it means that the sections to be evaluated in the action should be holonomic.
\end{Remark}

In view of the previous discussion, there are two crucial differences between a classical variational problem and a more general Griffiths variational problem that we would like to point out:
\begin{itemize}\itemsep=0pt
\item First of all, a classical variational problem (of first order) is formulated in a first order jet bundle, whereas a Griffiths variational problem can use any bundle in principle.
\item More important is the fact that the sections are integral for the set of forms~$\mathcal{I}$, that in the general case (as in the present article) could be different of the set of forms belonging to the contact structure.
\end{itemize}

A construction used in \cite{GotayCartan} becomes relevant to this article: Given a Griffiths variational problem $(p\colon W\to M,\lambda,\mathcal{I})$, it is possible to build another bundle $p'\colon W'\to M$ and a bundle map $q\colon W'\to W$ covering the identity on~$M$. The new bundle becomes part of a new variational problem $(p'\colon W'\to M,\lambda',0)$, and the extremals of the new variational problem can be set into a one-to-one correspondence with the extremals of the original variational problem through the map~$q$. The variational problem $(p'\colon W'\to M,\lambda',0)$ is called a \emph{Lepage-equivalent variational problem}.\footnote{There are some subtleties regarding the use of the word ``equivalent'' in this context; they will not be discussed here, so we are referring to the readers interested in these questions to the original article of Gotay.} The coordinates along the fibers of the map $p'\colon W'\to W$ can be seen as a generalization of momentum variables in the context of these variational problems. An important feature of a Lepage-equivalent variational problem is that it imposes no differential restrictions on the sections of $p'\colon W'\to M$, which could be useful when variations are performed.

A particular instance of this construction called \emph{canonical Lepage-equivalent problem} will be relevant for the present article. It requires that the differential constraints encoded by the exterior differential system $\mathcal{I}$ be generated by the sections of a subbundle $I\subset\wedge^\bullet (T^*W)$; under such conditions, we can define
\[
 W'|_w:=\big\{\lambda(w)+\beta\colon \beta\in I\cap\wedge^m (T^*_wW )\big\}\subset\wedge^m (T^*W )
\]
for all $w\in W$. The projection $p'\colon W'\to W$ is the restriction of the canonical map $\overline{\tau}^m_W\colon \wedge^m (T^*W)\to W$ to this subbundle, and the $m$-form $\lambda'$ is the pullback of the canonical $m$-form $\Theta\in\Omega^m (\wedge^m (T^*W))$ to~$W'$.

In order to describe the relevance of this construction, it will be necessary to refer to the concept of \emph{unified formalism for classical first order field theory}, as defined in \cite{1751-8121-40-40-005,2004JMP....45..360E,Prieto-Martinez2015203} and references therein. This formulation becomes useful when dealing with variational problems whose Lagrangian densities have singular Legendre transformations. The trick is to lift the variational problem to an space where both velocities and momenta are included and to forget about the differential restrictions on the fields imposed by the contact structure of the jet bundle. Therefore, variations can be performed without having to identify the independent degrees of freedom, and the formulas defining Legendre transform become part of the equations of motion; on the downside, these equations of motion are in general of the algebraic-differential type, and thus more difficult to deal with.
As we mentioned in the Introduction, this underlying bundle $W_{\text{cl}}$ for the unified formulation of field theory is given by equation~\eqref{eq:WClassic}; the Lagrangian functional~$\cL$ gives rise to an $m$-form $\Theta_\cL$ on $W_{\text{cl}}$, and the associated action becomes
\[
 \Gamma (W_{\text{cl}}\to M )\ni\sigma\mapsto\int_U\sigma^*\Theta_\cL.
\]

Finally, let us discuss briefly the relationship between Lepage-equivalent problems and the unified formalism. It can be seen that, when restricted to the particular case of the classical variational problem, the canonical Lepage-equivalent variational problem devised by Gotay reduces to the variational problem associated to the unified formalism. In fact, if the fields are sections of a bundle $\pi\colon E\to M$, the triple describing the associated Griffiths variational problem is
\[
 \big(\pi_1\colon J^1\pi\to M,\cL,\mathcal{I}_{{\rm con}}\big),
\]
where $\mathcal{I}_{{\rm con}}\subset\Omega^\bullet\big(J^1\tau\big)$ is the exterior differential system induced by the contact structure. In this case, we have the correspondences
\[
 W\leftrightsquigarrow J^1\pi,\qquad\lambda\leftrightsquigarrow\cL,\qquad\mathcal{I}\leftrightsquigarrow\mathcal{I}_{{\rm con}},\qquad W'\leftrightsquigarrow W_{\text{cl}},
\]
allowing us to think on Lepage-equivalent problems as generalizations of the unified formalism for field theory.

\section{Geometrical tools for Palatini gravity}\label{sec:LFT}

We choose to focus on variational problems of Griffiths type because there exists a description of Palatini gravity in terms of this kind of variational problems \cite{doi:10.1142/S0219887818500445}. The present section is devoted to give a brief account of the geometrical ingredients involved in this construction.

\subsection{Geometry of the jet space for the frame bundle}
\label{sec:geometry-jet-space}

The basic bundle is the frame bundle $\tau\colon LM\rightarrow M$ on the spacetime manifold $M$ ($\dim M=m$); because it is a principal bundle with structure group ${\rm GL}(m)$, we can lift this action to the jet bundle $J^1\tau$, so that we obtain a commutative diagram
\[
 \begin{tikzcd}[row sep=1.3cm,column sep=1.1cm]
 &
 J^1\tau
 \arrow[swap]{dl}{\tau_{10}}
 \arrow{dr}{p_{{\rm GL}(m)}^{J^1\tau}}
 \arrow{dd}{\tau_1}
 &
 \\
 LM
 \arrow[swap]{dr}{\tau}
 &
 &
 C(LM)
 \arrow{dl}{\overline{\tau}}
 \\
 &
 M
 &
 \end{tikzcd}
\]
where $C(LM):=J^1\tau/{\rm GL}(m)$ is the so called \emph{connection bundle of $LM$}, whose sections can be naturally identified with the principal connections of the bundle $\tau$ -- for details, see \cite{springerlink:10.1007/PL00004852} and references therein. It is interesting to note that there exists an affine isomorphism
\[
 F\colon \ J^1\tau\rightarrow LM\times_MC(LM)\colon \ j_x^1s\mapsto\big(s(x),\big[j_x^1s\big]_{{\rm GL}(m)}\big)
\]
and under this correspondence, the ${\rm GL}(m)$-action is isolated to the first factor in the product, namely
\begin{gather}\label{eq:IdentificationJetBundle}
 F\big(j_x^1s\cdot g\big)=\big(s(x)\cdot g,\big[j_x^1s\big]_{{\rm GL}(m)}\big).
\end{gather}
It means that a section of the bundle $\tau_1$ is equivalent to a connection on~$LM$ plus a moving frame $(X_1,\dots, X_m)$ on $M$; although this moving frame has no direct physical interpretation, we can associate a metric to it, namely, in contravariant terms,
\[
 g:=\eta^{ij}X_i\otimes X_j
\]
for some nondegenerate symmetric matrix $\eta$ -- see equation~\eqref{eq:EtaDefinition} below. It is the same to declare that the metric $g$ is the unique metric on $M$ making the moving frame $(X_1,\dots,X_m)$ (pseudo)orthonormal, with the signature given by $\eta$.

The tautological form $\widetilde{\theta}\in\Omega^1 (LM,\R^m )$ can be pulled back along $\tau_{10}$ to a $1$-form $\theta:=\tau_{10}^*\widetilde{\theta}$ on~$J^1\tau$; moreover, the Cartan form $\widetilde{\omega}\in\Omega^1\big(J^1\tau,V\tau\big)$, given by the formula
\[
\widetilde{\omega}|_{j_x^1s}:=T_{j_x^1s}\tau_{10}-T_xs\circ T_{j_x^1s}\tau_1,
\]
gives rise to a $\mathfrak{gl}(m)$-valued $1$-form $\omega$ on $J^1\tau$, by using the identification
\[
 V\tau\simeq LM\times\mathfrak{gl}(m).
\]
By means of the canonical basis $\{e_i\}$ on $\R^m$ and $\big\{E^i_j\big\}$ on $\mathfrak{gl}(m)$, where
\[
 \big(E_j^i\big)^q_p:=\delta^q_j\delta^i_p,
\]
we can define the collection of $1$-forms $\big\{\theta^i,\omega^i_j\big\}$ on $J^1\tau$ such that
\[
 \theta=\theta^ie_i,\qquad\omega=\omega ^j_iE^i_j.
\]
We also have the formula
\[
 \widetilde{\omega}=\omega ^j_i\big(E^i_j\big)_{J^1\tau},
\]
where $A_{J^1\tau}\in\mathfrak{X}^{Vp_{{\rm GL}(m)}^{J^1\tau}}\big(J^1\tau\big)$ is the infinitesimal generator associated to $A\in\mathfrak{gl}(m)$ for the lifted action. It can be proved that $\omega$ is a connection form for a principal connection on the bundle
\[
 p_{{\rm GL}(m)}^{J^1\tau}\colon \ J^1\tau\rightarrow C (LM ).
\]

\begin{Remark}[coordinates on the jet space $J^1\tau$ and the connection bundle $\overline{\tau}\colon C(LM)\to M$]\label{rem:CoordinatesJetConnections}
 Let $(\phi=(x^\mu),U)$ be a coordinate chart on $M$; for $u=(X_1,\dots,X_m)\in\tau^{-1}(U)$, we define the maps $u\mapsto e_i^\mu(u)$ such that
 \[
 X_i=e_i^\mu(u)\frac{\partial}{\partial x^\mu}
 \]
 for $1\leq i\leq m$. Then
 \[
 u\mapsto\big(x^\mu(\tau(u)),e^\mu_i(u)\big)
 \]
 defines a set of adapted coordinates on $\tau^{-1}(U)$; let $\big(x^\mu,e^\mu_i,e^\nu_{k\sigma}\big)$ be the associated coordinates on $\tau_1^{-1}(U)\subset J^1\tau$. We can also define the functions
 \[
 \Gamma_{\rho\sigma}^\mu:=-e^k_\rho e^\mu_{k\sigma},
 \]
 on $\tau_1^{-1}(U)$, and we can prove that they are ${\rm GL}(m)$-invariant; so, they define a set of coordinates on $\overline{\tau}^{-1}(U)\subset C(LM)$.
\end{Remark}

Let us define
\[
 \theta_0:=\theta^1\wedge\cdots\wedge\theta^m;
\]
as every $u\in LM$ is a collection $u=(X_1,\dots,X_m)$ of vectors on $\tau(u)\in M$, and $\theta^i$ is a $\tau_1$-horizontal $1$-form on $J^1\tau$, we can define the set of forms{\samepage
\[
\theta_{i_1\cdots i_k}|_{j_x^1s}:=X_{i_1}\lrcorner\cdots X_{i_k}\lrcorner \theta_0|_{j_x^1s}
\]
for $1\leq i_1,\dots,i_k\leq m$, where $j_x^1s\in J^1\tau$ is any element such that $u=\tau_{10}\big(j_x^1s\big)$.}

Although the reduction scheme we will develop in this article would work for any signature, let us fix it using the matrix
\begin{gather}\label{eq:EtaDefinition}
 \eta:=
 \begin{bmatrix}
 -1&0&\cdots&0\\
 0&1&&0\\
 \vdots&&\ddots&\vdots\\
 0&\cdots&0&1
 \end{bmatrix}\in {\rm GL}(m).
\end{gather}
Accordingly, let $\eta_{ij}$ be its $(i,j)$-entry; we will represent with the symbol $\eta^{ij}$ the $\left(i,j\right)$-entry of its inverse. With these ingredients we can construct the \emph{Palatini Lagrangian}
\begin{gather}\label{eq:PalatiniLagrangianInvariant}
 \cL_{\rm PG}:=\eta^{ip}\theta_{ik}\wedge\Omega^k_p,
\end{gather}
where $\Omega:=\Omega^i_jE^j_i$ is the curvature of the canonical connection $\omega$. This $m$-form will determine the dynamics of the vacuum gravity in this formulation.

Finally, let us describe a decomposition of $\mathfrak{gl}(m)$ induced by $\eta$. In fact, this matrix yields to a real form $\uf$ in $\mathfrak{gl}\left(m,\C\right)$, given by
\[
 \uf=\big\{\xi\in\mathfrak{gl} (m,\C )\colon \xi^\dagger\eta+\eta\xi=0\big\}
\]
and thus we have a Cartan decomposition
\[
 \mathfrak{gl} (m,\C )=\uf\oplus\mathfrak{s}.
\]
Given the inclusion
\[
 \mathfrak{gl}(m)\subset\mathfrak{gl} (m,\C),
\]
we obtain the decomposition
\[
 \mathfrak{gl}(m)=\kf\oplus\pf.
\]
The subalgebra $\kf$ is the Lie algebra of the subgroup $K\subset {\rm GL}(m)$, composed of the linear transformations keeping invariant the matrix~$\eta$,
\[
 K:=\big\{A=\big(A_i^j\big)\colon \eta_{ij}A^i_kA^j_l=\eta_{kl}\big\}.
\]
The canonical action of ${\rm GL}(m)$ on $LM$ restricts to an free action of $K$ on this bundle; accordingly, we have the $K$-principal bundle
\[
 p_K^{LM}\colon \ LM\to\Sigma:=LM/K.
\]
The bundle $\tau_\Sigma\colon \Sigma\to M$ induced by this quotient has an immediate physical meaning: Any section $\zeta\colon M\to\Sigma$ is a metric on~$M$ with $\eta$-signature. In fact, let us define the map $q\colon LM\to TM\otimes_MTM$ via
\[
 q (X_1,\dots,X_m ):=\eta^{ij}X_i\otimes X_j;
\]
indicating by $R_k\colon LM\to LM$ the action of an element $k\in K$ on $LM$, it can be proved that $q\circ R_k=q$, so that there exists a map $\overline{q}\colon LM/K\to TM\otimes_MTM$ making the following diagram commutative
\[
 \begin{tikzcd}[row sep=1.3cm,column sep=2.1cm]
 LM
 \arrow{r}{p_K^{LM}}
 \arrow[swap]{dr}{q}
 &
 \Sigma
 \arrow{d}{\overline{q}}
 \\
 &
 TM\otimes_MTM
 \end{tikzcd}
\]
When restricted to the subbundle of nondegenerate symmetric $2$-tensors (with $\eta$-signature) on~$M$, this map becomes a bundle isomorphism on~$M$, allowing us to interpret~$\Sigma$ as the bundle of metrics with $\eta$-signature on~$M$.

\begin{Remark}[adapted coordinates for the principal bundle $p_K^{LM}\colon LM\to\Sigma$]\label{rem:CoordinatesForMetrics}
 It will be useful to introduce a set of coordinates on $LM$ adapted to the quotent map~$p_K^{LM}$. In order to proceed, we need a theorem on generalized polar decomposition~\cite{higham2003j}. So, let $U\subset M$ be a parallelizable open set in $M$, let $\{Z_1,\dots,Z_m\}\subset\mathfrak{X}(U)$ be a moving basis on $U$, and $\{e_1,\dots,e_n\}$ the canonical basis on $\mathbb{R}^m$; for every $x\in U$ denote $V=T_xM$ and $W=\mathbb{R}^m$. The matrix $\eta$ defines a linear map
 \[
 \eta\colon \ W\to W^*;
 \]
 another operator that will be important for the formulation of the polar decomposition theorem is
 \[
 I\colon \ V\to W\colon \ Z_\mu\mapsto\delta_\mu^ie_i.
 \]
 Also, set
 \[
 \widetilde{\eta}:=I^*\circ\eta\circ I.
 \]
 It is clear that, in this setting, a metric is a map
 \[
 g_x\colon \ V\to V^*
 \]
 for every $x\in U$, and a vielbein becomes a linear map
 \[
 E\colon \ V\to W.
 \]
 Then, the following can be proved.
 \begin{Theorem}[generalized polar decomposition]
 Let $g\colon V\to V^*$ be an invertible linear map such that $\widetilde{\eta}\circ g^{-1}$ has no eigenvalues on the nonpositive real axis. Then $g$ can be uniquely descomposed as
 \[
 g=Q\circ s,
 \]
 where $Q\colon W\to V^*$ satisfies
 \[
 Q^*\circ\widetilde{\eta}\circ Q=\eta
 \]
 and for $s\colon V\to W$ the following property
 \[
 I^*\circ\eta\circ s= (I^*\circ\eta\circ s )^*
 \]
 holds.
 \end{Theorem}
 Therefore, the condition on the factor $Q$ tells us that the composite map
 \[
 I^*\circ Q\colon \ W\to W^*
 \]
 belongs to $K$. Thus, let $V\subset\tau_\Sigma^{-1}(U)$ be an open set in~$\Sigma$ such that $\widetilde{\eta}\circ g^{-1}$ has no eigenvalues on the nonpositive real axis for every $g\in V$. Using this theorem, we can define the $K$-bundle isomorphism
 \[
 u\in\big(p_K^{LM}\big)^{-1}(V)\mapsto\big(p_K^{LM}(u),I^*\circ Q\big)\in V\times K,
 \]
 where $Q$ is the corresponding factor in the polar decomposition of $g=p_K^{LM}(u)$. Finally, given a~set of coordinates $\big(x^\mu,g^{\mu\nu}\big)$ on $\tau_\Sigma^{-1}(U)$ and $\big(k_{i}^j\big)$ on~$K$, we can introduce the set of coordinates
 \[
 u\mapsto\big(x^{\mu},g^{\mu\nu},k_i^j\big)
 \]
 on $\tau_1^{-1}(U)\subset LM$, where $\big(x^\mu,g^{\mu\nu}\big)$ are the coordinates for~$p_K^{LM}(u)$ and $\big(k_i^j\big)$ are the coordinates on~$K$ for~$I^*\circ Q$.
\end{Remark}

\subsection[Restrictions in Palatini gravity: Zero torsion submanifold and metricity forms]{Restrictions in Palatini gravity: Zero torsion submanifold\\ and metricity forms}\label{sec:zero-tors-subm}

It is time to discuss the restrictions we must impose on the sections of~$\tau_1$ in order to have a~characterization of a gravity field in this description. Our aim is to describe a metric and a~connection on the spacetime, and the restrictions to be considered will establish the relationship between them; this approach has been extensively discussed in the references \cite{capriotti14:_differ_palat,doi:10.1142/S0219887818500445}.

Recall that, according to the identification performed by the map~$F$~-- see equation~\eqref{eq:IdentificationJetBundle} above~-- a section of $J^1\tau$ can be seen as a pair composed by a frame plus a connection. Correspondingly, there are two types of conditions to be imposed to a section of~$J^1\tau$, each of them motivated on physical grounds which we will not discuss here:
\begin{enumerate}\itemsep=0pt
\item[1)] the connection which is a solution for the field equations of Palatini gravity must be torsionless, and
\item[2)] this connection must be metric for the solution metric.
\end{enumerate}

The canonical forms defined in the previous section allow us to set the \emph{torsion form}
\[
 T:=\big({\rm d}\theta^j+\omega^j_k\wedge\theta^k\big)\otimes e_j\in\Omega^2\big(J^1\tau,\R^m\big).
\]
Now, every connection $\Gamma\colon M\rightarrow C(LM)$ gives rise to a section $\sigma_\Gamma\colon LM\rightarrow J^1\tau$ of the bundle $\tau_{10}\colon J^1\tau\rightarrow LM$, as the equivariance of the following diagram shows
\[
 \begin{tikzcd}[row sep=1.3cm,column sep=1.1cm]
 &
 J^1\tau
 \arrow{dl}{\tau_{10}}
 \arrow{dr}{p_{{\rm GL}(m)}^{J^1\tau}}
 &
 \\
 LM
 \arrow{dr}{\tau}
 \arrow[dashed,bend left=45]{ur}{\sigma_\Gamma}
 &
 &
 C(LM)
 \arrow[swap]{dl}{\overline{\tau}}
 \\
 &
 M
 \arrow[dashed,bend right=45,swap]{ur}{\Gamma}
 &
 \end{tikzcd}
\]
The interesting fact is that the pullback form $\sigma_\Gamma^*T$ coincides with the torsion of the connection~$\Gamma$. Additionally, it can be proved that $T$ is a $1$-horizontal form on $\tau_1\colon J^1\tau\to M$, so that there exists a maximal (respect to the inclusion) submanifold $i_0\colon \cT_0\hookrightarrow J^1\tau$ such that
\begin{enumerate}\itemsep=0pt
\item[1)] $\cT_0$ is transversal to the fibers of $\tau_1\colon J^1\tau\rightarrow M$ (namely, $T_{j^1_xs} (\cT_0)\oplus V_{j_x^1s}\tau_1=T_{j_x^1s}\left(J^1\tau\right)$), and
\item[2)] it annihilates the torsion, i.e.,
 \[
 i_0^*T=0.
 \]
\end{enumerate}
The transformation properties of the form $T$ allow us to conclude that $\cT_0$ is ${\rm GL}(m)$-invariant. The connections associated to sections of $J^1\tau$ taking values in $\cT_0$ are torsionless, so that the zero torsion restriction can be achieved through the requirement that these sections take values in this submanifold. Accordingly, we can use the affine isomorphism $F\colon J^1\tau\to LM\times C(LM)$, to define the \emph{bundle of torsionless connections} as the bundle $\overline{\tau}'\colon C_0(LM)\to M$ obtained by restricting~$F$ to~$\cT_0$
\[
 C_0(LM):={\rm pr}_2 (F (\cT_0 ) ).
\]
Moreover, the following lemma can be proved using standard facts about principal bundles \cite{KN1}.

\begin{Lemma}\label{lem:torsion-zero-bundle}
 The submanifold $\cT_0\subset J^1\tau$ is a subbundle of the ${\rm GL}(m)$-bundle $p_{{\rm GL}(m)}^{J^1\tau}\colon J^1\tau\to C(LM)$, associated to the isomorphism ${\rm id}\colon {\rm GL}(m)\to {\rm GL}(m)$. Also, it is a ${\rm GL}(m)$-principal bundle with respect to the restriction of the ${\rm GL}(m)$-action.
\end{Lemma}
These considerations give rise to the commutative diagram
\[
 \begin{tikzcd}[row sep=1.3cm,column sep=1.1cm]
 &
 \cT_0
 \arrow[swap]{dl}{\tau_{10}'}
 \arrow{dr}{p_{{\rm GL}(m)}^{\cT_0}= p_{{\rm GL}(m)}^{J^1\tau}\big|_{\cT_0}}
 &
 \\
 LM
 \arrow{dr}{\tau'}
 &
 &
 C_0(LM)
 \arrow{dl}{\overline{\tau}'}
 \\
 &
 M
 &
 \end{tikzcd}
\]

\begin{Remark}[local description for $\cT_0$]
 Let $\left(x^\mu,e_k^\nu\right)$ be a set of adapted coordinates for $LM$ induced on $\tau^{-1}(U)$ by a set of coordinates $ (x^\mu)$ on $U\subset M$; as was shown in Remark~\ref{rem:CoordinatesJetConnections}, it induces coordinates $ (x^\mu,e_k^\nu,e^\nu_{k\sigma} )$ on $\tau_1^{-1}(U)$. On this open set we have
 \[
 T=e_\sigma^i\big(e^k_\mu e_{k\nu}^\sigma dx^\mu\wedge dx^\nu\big)\otimes e_i,
 \]
 where $(e^k_\mu)$ is the inverse matrix of $(e_k^\mu)$, so that the set $\cT_0\cap\tau_1^{-1}(U)$ is described by the constraints
 \[
 e^k_\mu e_{k\nu}^\sigma=e^k_\nu e_{k\mu}^\sigma.
 \]
\end{Remark}

On the other hand, the metricity condition has a differential nature: As we mentioned above, matrix $\eta$ determines a factorization of $\mathfrak{gl}(m)$ in a subalgebra $\kf$ (the subalgebra of $\eta$-Lorentz transformations) and an invariant subspace $\pf$. The explicit formulas for this decomposition are given by the projectors
\[
 A_\kf:=\frac{1}{2}\big(A-\eta A^T\eta\big),\qquad A_\pf:=\frac{1}{2}\big(A+\eta A^T\eta\big)
\]
for every $A\in\mathfrak{gl}(m)$. The metricity condition is imposed on a section $s\colon M\rightarrow J^1\tau$ by requiring that
\begin{gather}\label{eq:MetricityConds}
 s^*\omega_\pf=0,
\end{gather}
where $\omega_\pf$ is the $\pf$-component of the canonical connection $\omega$ respect to this decomposition. Taking into account the affine isomorphism $F\colon J^1\tau\to LM\times_MC(LM)$, this constraint means that the parallel transport of the connection ${\rm pr}_2\circ F\circ s\colon M\to C(LM)$ leaves invariant the metric associated to the vielbein ${\rm pr}_1\circ F\circ s\colon M\to LM$ (see equation~\eqref{eq:GammaInLocalTerms} below).

\begin{Remark}[local expression for the metricity conditions] Let us provide a description of the metricity conditions \eqref{eq:MetricityConds} in terms of the adapted coordinates $(x^\mu,e^\mu_k,e^\mu_{k\nu})$ on $\tau_1^{-1}(U)\subset J^1\tau$. It results that for a section
 \[
 s\colon \ (x^\mu )\mapsto\big(x^\mu,e^\mu_k(x),e^\mu_{k\nu}(x)\big),
 \]
 these conditions can be written as \cite{capriotti14:_differ_palat}
 \begin{equation}\label{eq:MetricityLocal}
 {\rm d} {g}^{\mu\nu}+\big({g}^{\mu\sigma}\Gamma^\nu_{\gamma\sigma}+{g}^{\nu\sigma}\Gamma^\mu_{\gamma\sigma}\big) {\rm d}x^\gamma=0,
 \end{equation}
 where ${g}^{\mu\nu}(x):=\eta^{kl}e_k^\mu(x)e_l^\nu(x)$ is the metric associated to this section and
 \[
 \Gamma^\mu_{\rho\sigma}(x):=-e^k_\rho(x)e^\mu_{k\sigma}(x)
 \]
 are the Christoffel symbols for the associated connection. As it is well-known, these conditions plus zero torsion imply the metric and the symbols become related by Levi-Civita formula
 \[
 \Gamma_{\mu\nu}^\sigma=\frac{1}{2}g^{\sigma\rho}\left(\frac{\partial g_{\rho\mu}}{\partial x^\nu}+\frac{\partial g_{\rho\nu}}{\partial x^\mu}-\frac{\partial g_{\mu\nu}}{\partial x^\rho}\right).
 \]
\end{Remark}

\begin{Remark}[on the nature of the metricity conditions] In some approaches to Palatini gravity, it is customary to work with a ${\rm SO}(3,1)$ connection, without a clear reference to the bundle to which this connection should be associated. It gives rise to an ambiguity that could be problematic in some cases. Namely, when in local terms it is referred to a $\mathfrak{so}(3,1)$-valued $1$-form
 \[
 \alpha\colon \ TM\to\mathfrak{so} (3,1),
 \]
 it is not clear whether it is a local version of a connection on a principal bundle with structure group ${\rm SO}(3,1)$ or instead, a local version of a connection on a principal bundle with a larger structure group $G\supset {\rm SO}(3,1)$ taking values in the smaller Lie algebra $\mathfrak{so}(3,1)\subset\g$. It is not uncommon to find articles in the literature where the information provided to the reader to decide in which case one is working is not enough, because in general, the underlying principal bundle is not associated to quantities of physical nature. It should remain clear that in general, these connections are not equivalent, unless a reduction of the $G$-bundle to a subbundle with structure group ${\rm SO}(3,1)$ is admitted. When $G={\rm GL}(4)$, it is equivalent to have a metric with signature $3+1$; in this case, the ${\rm SO}(3,1)$-bundle is the bundle of pseudo orthogonal frames. It could be argued that to have a metric at our disposal is a feature of working with metric-affine theories of gravitation; but then it should be pointed out in this case that variations of the metric have the undesired side effect of changing the ${\rm SO}(3,1)$-bundle. In order to avoid this behavior, we have adopted metricity conditions \eqref{eq:MetricityConds}: A connection verifying these conditions will be a connection on the frame bundle, whose structure group is ${\rm GL}(m)$, and will have a local version taking values in the subalgebra $\kf\subset\mathfrak{gl}(m)$ associated to the matrix~$\eta$.
\end{Remark}

\section{Griffiths variational problem for Palatini gravity}\label{sec:vari-probl-palat}

The variational problem we will consider here for the Palatini gravity is not a classical one; it will differ from a variational problem of this kind in both of the aspects mentioned in Section~\ref{sec:vari-probl-unif}:
\begin{itemize}\itemsep=0pt
\item The relevant bundle is not the first order jet $J^1\tau$; instead, it is the subset $\cT_0$ consisting in the jets associated to torsionless connections. Due to this fact, we will consider the pullback of the canonical forms and the restriction of maps from $J^1\tau$ to $\cT_0$; unless explicitly stated, the new forms and maps will be indicated with the same symbols. An exception to this rule will be the restriction of the bundle maps $\tau_{10}$ and $\tau_1$, which will be indicated as
 \[
 \tau_{10}'\colon \ \cT_0\to LM\qquad\text{and}\qquad\tau_1'\colon \ \cT_0\to M.
 \]

\item The forms we will use for the restriction of the sections of
 $\tau_1'\colon \cT_0\rightarrow M$ are not the whole set of contact forms
 $\big\{\omega_j^i\big\}$, but a geometrically relevant subset,
 namely, the components of the metricity forms $\omega_\pf$.
\end{itemize}

Using these considerations, we will introduce the following definition.
\begin{Definition}[Griffiths variational problem for Palatini gravity]\label{def:VariationalProblemPalatini}
 The variational problem for Palatini gravity is given by the action
 \[
 S [\sigma ]:=\int_U\sigma^*\cL_{\rm PG},
 \]
 where $\sigma\colon U\subset M\rightarrow \cT_0$ is any section of $\tau_1'$ such that $\sigma^*\omega_\pf=0$. According to equation~\eqref{rem:griffiths-notation}, it is described by the triple $(\tau_1'\colon \cT_0\to M,\cL_{\text{\rm PG}},\langle \omega_\pf\rangle)$, where $\langle \cdot\rangle$ indicates the exterior differential system generated by the set of forms enclosed by the brackets.
\end{Definition}

\begin{Remark}\label{rem:restr-palat-grav}
 The considerations made in Section \ref{sec:LFT} allow us to compare the variational problem given by Definition~\ref{def:VariationalProblemPalatini} with the classical variational problem, as in \cite{CASTRILLONLOPEZ2013352}. In fact, Proposition~4 in this reference tells us that when holonomic sections are factorized using the isomorphism~$F$ (in~\cite{CASTRILLONLOPEZ2013352} it is called~$\Psi$), the connection
 \[
 {\rm pr}_2\circ F\circ s\colon \ M\to C (LM )
 \]
 is flat, a restriction far more stringent than those imposed by the metricity conditions, that only require the associated connection form to have values in~$\kf$. For a description of these admissible variations in the case of the metricity constraints, see \cite[Section~4.1]{Capriotti:2019bqh}.

 It is interesting to point out the effect that the replacement of the contact structure with the differential system $\langle\omega_\pf \rangle$ has on the way that Euler--Lagrange equations are calculated. The key change has to do with the characterization of the admissible infinitesimal variations: Given that the metricity constraints are a part of the constraints imposed by the contact structure, it results that the set of admissible infinitesimal variations considered in~\cite{CASTRILLONLOPEZ2013352} contain them as a~subset.
\end{Remark}

In order to establish the unified version of the equations of motion for Palatini gravity, it will be necessary to represent the metricity constraints as a subbundle of the set of $m$-forms on $\cT_0$. To this end, let us define the \emph{metricity subbundle} $I_{\rm PG}^m$ on $\cT_0$,
\[
 I_{\rm PG}^m:=\big\{\eta^{ik}\beta_{kp}\wedge\omega_k^p\colon \beta_{ij}\in\wedge^{m-1}_1 (T^*\cT_0),\beta_{ij}-\beta_{ji}=0\big\}\subset\wedge^m_2 (T^*\cT_0),
\]
where $\wedge^{m-1}_1 (T^*\cT_0)$ indicates the set of $\tau_1'$-horizontal $(m-1)$-covectors on $\cT_0$.

\begin{Remark}[local description for the metricity subbundle]\label{rem:LocalDescriptionMetricity}
 Given a nondegenerate symmetric matrix $(g^{\mu\nu})$, there exists a matrix $\left(e^\mu_k\right)$ such that
 \[
 g^{\mu\nu}=\eta^{kl}e_k^\mu e_l^\nu;
 \]
 also, any two matrices $(e^\mu_k)$, $(f^\mu_k)$ giving rise to the same symmetric matrix $\left(g^{\mu\nu}\right)$ differ in an element $(k^i_j)\in {\rm GL}(m)$ such that
 \[
 k^i_jk^l_m\eta^{jm}=\eta^{il}.
 \]
 Using the factorization
 \[
 \cT_0=LM\times_M C_0(LM),
 \]
 we can set the coordinates $\big(x^\mu,g^{\mu\nu},k^i_j,\Gamma^\mu_{\rho\sigma}\big)$ on $\cT_0$; the functions $\Gamma^\mu_{\rho\sigma}$ are not independent; rather, they are constrained by the relations
 \[
 \Gamma^\mu_{\rho\sigma}=\Gamma^\mu_{\sigma\rho}.
 \]
 On the other hand, metricity constraints are given by equation~\eqref{eq:MetricityLocal}); therefore
 \[
 I_{\rm PG}^m\big|_{(x^\mu,g^{\mu\nu},k^i_j,\Gamma^\mu_{\rho\sigma})}=\big\{p_{\mu\nu}^\rho \big[{\rm d}{g}^{\mu\nu}+\big({g}^{\mu\sigma}\Gamma^\nu_{\gamma\sigma}+{g}^{\nu\sigma}\Gamma^\mu_{\gamma\sigma}\big) {\rm d}x^\gamma\big]\wedge\eta_\rho\colon p_{\mu\nu}^\rho=p_{\nu\mu}^\rho\big\},
 \]
 where, as before
 \[
 \eta_\rho:=\frac{\partial}{\partial x^\rho}\lrcorner\eta,\qquad\eta={\rm d}x^1\wedge\cdots\wedge {\rm d}x^m.
 \]
 It is a description of the fibers of $I_{\rm PG}^m$ in local terms; the independence of these forms respect to the variables $\big(k_i^j\big)$ is due to the fact that the group $K$ describes a fundamental symmetry of this description of gravity.
\end{Remark}

\begin{Remark}[metricity conditions and the annihilation of the metricity tensor]\label{rem:MetricityTensorAndConds}
It could be useful to give an explanation of the relevance of this subbundle. To this end, recall the notion of \emph{metricity tensor of a connection respect to a metric}~\cite{Hehl19951}; in short, given a manifold with a~metric and a~connection in it, this tensor is essentially the covariant derivative of the metric. When working in the jet space of the frame bundle, where a section induced both a metric and a connection, this tensor can be encoded in terms of the canonical connection. Concretely, let us suppose that $s\colon M\to\cT_0$ is a section with the following property
\[
 s^*\big(\eta^{ik}\beta_{ip}\wedge\omega^p_k\big)=0
\]
for any collection of $l$-forms $\beta_{ip}$ such that
\[
 \beta_{ip}-\beta_{pi}=0.
\]
Then the connection on $LM$ induced by $s$ has zero metricity tensor; therefore, the metricity subbundle allows us to introduce the set of restrictions imposed by the annihilation of the metricity tensor into the Lagrangian. From this viewpoint, forms $\beta_{ij}$ in the definition of the metricity subbundle play the role of Lagrange multipliers.
\end{Remark}

With the metricity subbundle in mind, and following the prescriptions made in \cite{GotayCartan} for the construction of the canonical Lepage-equivalent variational problem, we define the affine subbundle
\begin{gather}\label{eq:BundleOfFormsOnT0}
 W_{\rm PG}:=\cL_{\rm PG}+I_{\rm PG}^m\subset\wedge^m_2 (T^*\cT_0 ),
\end{gather}
which comes with the projection
\[
 \tau_{\rm PG}\colon \ W_{\rm PG}\rightarrow \cT_0\colon \ \alpha\in\wedge^m_2\big(T^*_{j^1_xs}\cT_0\big)\mapsto j_x^1s.
\]
We are referring to the interested reader to the article of Gotay for details; for the purposes of the present work, it should be enough to mention that the bundle defined by equation~\eqref{eq:BundleOfFormsOnT0} will become the underlying bundle for the Lepage-equivalent problem associated to the Griffiths variational problem describing Palatini gravity (see Definition~\ref{def:VariationalProblemPalatini} above).

Because this is a subbundle in the set of $m$-forms on $\cT_0$, it has a canonical $m$-form $\lambda_{\rm PG}$ on it given by
\[
\lambda_{\rm PG}|_\alpha(w_1,\dots,w_m):=\alpha(T_\alpha\tau_{\rm PG}(w_1),\dots,T_\alpha\tau_{\rm PG}(w_m)),\qquad w_1,\dots,w_m\in T_\alpha(W_{\rm PG}).
\]
This $m$-form will be the Lagrangian form for the Lepage-equivalent problem associated to the Griffiths variational problem for Palatini gravity, namely
\begin{equation}\label{eq:LepEquivPalatini}
 (\tau_1'\circ\tau_{\rm PG}\colon W_{\rm PG}\to M,\lambda_{\rm PG},0 ).
\end{equation}
Because no differential restrictions must be considered when performing variations of the variational problem
\[
 \Gamma\left(\tau_1'\circ\tau_{\rm PG}\right)\ni\Gamma\mapsto\int_U\Gamma^*\lambda_{\rm PG},
\]
and assuming that the variations annihilate at the boundary, we obtain that the Euler-Lagrange equations for the Lepage-equivalent variational problem \eqref{eq:LepEquivPalatini} are
\[
 \Gamma^* (X\lrcorner {\rm d}\lambda_{\rm PG} )=0\qquad\text{for all}\qquad X\in\mathfrak{X}^{V(\tau_1'\circ\tau_{\rm PG})}(W_{\rm PG}).
\]
The relevance of the unified formalism in dealing with the variational problems posed by Definition \ref{def:VariationalProblemPalatini} is guaranteed by the following result \cite{doi:10.1142/S0219887818500445}.
\begin{Proposition}\label{prop:FieldTheoryEqsWL}
 A section $s\colon U\subset M\rightarrow \cT_0$ is critical for the variational problem established in Definition~{\rm \ref{def:VariationalProblemPalatini}} if and only if there exists a section $\Gamma\colon U\subset M\rightarrow W_{\rm PG}$ such that
 \begin{enumerate}\itemsep=0pt
 \item[$1)$] $\Gamma$ covers $s$, i.e., $\tau_{\rm PG}\circ\Gamma=s$, and
 \item[$2)$] $\Gamma^*(X\lrcorner {\rm d}\lambda_{\rm PG})=0$, for all $X\in\mathfrak{X}^{V(\tau_1'\circ\tau_{\rm PG})}(W_{\rm PG})$.
 \end{enumerate}
\end{Proposition}
$\Gamma$ is called a \emph{solution} of the Palatini gravity equations of motion.

\begin{Remark} Although the proof in \cite{doi:10.1142/S0219887818500445} refers to sections of $\tau_1\colon J^1\tau\to M$, it can be also {readily adapted} to cover this case; in this regard, see Appendix~\ref{app:LiftToTorsionZero}.
\end{Remark}

The situation described by Proposition~\ref{prop:FieldTheoryEqsWL} is summarized in the following diagram:
\[
 \begin{tikzcd}[row sep=1.3cm,column sep=1.1cm]
 W_{\rm PG}
 \arrow{dr}{\tau_1'\circ\tau_{\rm PG}}
 \arrow{rr}{\tau_{\rm PG}}
 &
 &
 \cT_0
 \arrow[swap]{dl}{\tau_1'}
 \\
 &
 M
 \arrow[dashed,bend left=40]{ul}{\Gamma}
 \arrow[dashed,bend right=40,swap]{ur}{s}
 &
 \end{tikzcd}
\]
Accordingly, any section $s$ that is extremal for the Griffiths variational problem
\[
 (\tau_1'\colon \cT_0\to M,\cL_{\text{\rm PG}}, \langle\omega_\pf \rangle )
\]
can be lifted to a section $\Gamma$ that is extremal for the Lepage-equivalent variational problem
\[
 (\tau_1'\circ\tau_{\rm PG}\colon W_{\rm PG}\to M,\lambda_{\rm PG},0 )
\]
and viceversa.

\begin{Remark}[local version of Proposition~\ref{prop:FieldTheoryEqsWL}] According to Remark~\ref{rem:LocalDescriptionMetricity}, section $\Gamma$ will be obtained from section~$s$ by providing the set of functions
 \[
 (p_{\mu\nu}^\sigma )\colon \ U\to\mathbb{R}^{\frac{m^2(m+1)}{2}}.
 \]
 These momentum variables $p_{\mu\nu}^\sigma$ are determined by the equations of motion in Proposition~\ref{prop:FieldTheoryEqsWL}, concretely, through the contractions
 \[
 \Gamma^*\left(\frac{\partial}{\partial\Gamma^\sigma_{\mu\nu}}\lrcorner {\rm d}\lambda_{\rm PG}\right)=0.
 \]
\end{Remark}

We will see below (Section~\ref{sec:first-order-vari}) that an unified formalism for (first order) Einstein--Hilbert variational problem can also be given; the reduction and reconstruction theorems (see Section~\ref{sec:einst-hilb-grav}) will be proved using these lifted systems.

\section{Symmetry and reduction}\label{sec:symmetry-momentum}

As we described in the introduction, a crucial ingredient in Routh reduction is the restriction of the dynamics to a level set of the momentum map. It forces us to discuss the presence of natural symmetries in our formulation of gravity, and to construct their momentum maps. Also, this procedure requires the choice of a connection on a bundle obtained by quotient out the symmetries of the variational problem. This section is devoted to these tasks.

\subsection{Momentum map and connection}
\label{sec:momentum-map}

As we said above (see Lemma \ref{lem:torsion-zero-bundle}), there exists a ${\rm GL}(m)$-action on $\cT_0$; nevertheless, the Lagrangian $\cL_{\rm PG}$ is preserved by the action of the subgroup $K\subset {\rm GL}(m)$ composed of the linear transformations keeping invariant the matrix $\eta$,
\[
 K:=\big\{A=\big(A_i^j\big)\colon \eta_{ij}A^i_kA^j_l=\eta_{kl}\big\}.
\]
We can lift the ${\rm GL}(m)$-action to $\wedge^m (\cT_0)$; it results that the subbundle~$I^m_{\rm PG}$ is also preserved by the action of~$K$, and so
\[
 K\cdot W_{\rm PG}\subset W_{\rm PG}.
\]
Our aim is to find a momentum map for this action, in the sense of the following definition.

\begin{Definition}
A \emph{momentum map} for the action of $K$ on $W_{\rm PG}$ is a map
\[
 J\colon \ W_{\rm PG}\to \Lambda^{m-1} (T^*W_{\rm PG})\otimes\kf^*
\]
over the identity in $W_{\rm PG}$ such that
\[
\xi_{W_{\rm PG}}\lrcorner {\rm d}\lambda_{\rm PG}=-{\rm d}J_\xi,
\]
where $J_\xi$ is the $(m-1)$-form on $W_{\rm PG}$ whose value at $\alpha\in W_{\rm PG}$ is $J_\xi(\alpha)=\langle J(\alpha),\xi \rangle$.

A momentum map is $Ad^*$-\emph{equivariant} if it satisfies
\[
 \langle J(g\alpha),{\rm Ad}_{g^{-1}}\xi \rangle=g \langle J(\alpha),\xi \rangle.
\]
Also, it is said that $J$ is \emph{conserved along a section $\Gamma\colon M\to W_{\rm PG}$} if and only if $\Gamma^* ({\rm d}J)=0$.
\end{Definition}
\begin{Remark}\label{rem:moment-map-conserved} A clarification about the nomenclature seems necessary here: Suppose $M=\mathbb{R}\times N$ for some $(m-1)$-dimensional manifold $N$, we can take $U\subset N$ and $\Gamma\colon [T_1,T_2 ]\times U\to W_{\rm PG}$ such that $J$ is conserved along $\Gamma$ and moreover
 \[
 J (\Gamma (t,q ) )=0\qquad\text{for all} \quad T_1\leq t\leq T_2\quad \text{and} \quad q\in\partial U;
 \]
 then, denoting $\Gamma_t\colon U\to W_{\rm PG}\colon r\mapsto\Gamma(t,r)$ for any $t\in[T_1,T_2]$, Stokes' theorem will tell us that
 \[
 \int_U\Gamma_{T_1}^*J=\int_U\Gamma_{T_2}^*J.
 \]
 It is in this sense that the momentum map is ``conserved'', namely, the integral of the $(m-1)$-form $\Gamma_t^*J$ is independent of $t$.
\end{Remark}

Thus, we obtain Noether's theorem in this setting:
\begin{Proposition} The momentum map $J$ is conserved along solutions of the Palatini gravity equations of motion.
\end{Proposition}
\begin{proof}
 Recall that $\Gamma\colon U\subset M\rightarrow W_{\rm PG}$ is a solution for the Palatini gravity equations of motion if and only if
 \[
 \Gamma^*(Z\lrcorner {\rm d} \lambda_{\rm PG})=0
 \]
 for any $\tau_1'\circ\tau_{\rm PG}$-vertical vector field $Z$. Then for each $\xi\in\kf$ we have
\[
{\rm d}(\Gamma^*J_\xi)=\Gamma^*({\rm d}J_\xi)=\Gamma^*(-\xi_{W_{\rm PG}}\lrcorner {\rm d} \lambda_{\rm PG})=0,
\]
and therefore the momentum is conserved along solutions.
\end{proof}

Accordingly, we think of a ``momentum'' $\widehat{\mu}$ as an element $\widehat{\mu}\in \Omega^{m-1}(W_{\rm PG},\mathfrak{gl}(m)^*)$, i.e., as a~$\mathfrak{gl}(m)^*$-valued $(m-1)$-form on $W_{\rm PG}$; a conserved value $\widehat{\mu}$ of the momentum map is a~closed one, i.e., ${\rm d}\widehat{\mu}=0$.

The construction of a momentum map for the action on $W_{\rm PG}$ is standard~\cite{Gotay:1997eg}:
\begin{Lemma}\label{lem:mmap} The map $J\colon W_{\rm PG}\to \Lambda^{m-1} (T^*W_{\rm PG})\otimes \kf^*$ defined by
\[
\langle J(\alpha),\xi \rangle=\xi_{W_{\rm PG}}(\alpha) \lrcorner \lambda_{\rm PG} |_\alpha,
\]
for each $\xi\in\kf$, is an ${\rm Ad}^*$-equivariant momentum map for the ${\rm GL}(m)$-action on $W_{\rm PG}$.
\end{Lemma}

Now, because
\[
 T\tau_{\rm PG}\circ\xi_{W_{\rm PG}}=\xi_{\cT_0}\circ\tau_{\rm PG},
\]
then for every $\alpha\in W_{\rm PG}$ it results that
\begin{align*}
 \langle J(\alpha),\xi \rangle&=\xi_{W_{\rm PG}}(\alpha) \lrcorner \lambda_{\rm PG} |_\alpha
 =\xi_{\cT_0}\lrcorner\alpha\\
 &=i_0^*\big[\xi_{J^1\tau}\lrcorner\big(\eta^{ip}\theta_{pk}\wedge\Omega^k_i+ \eta^{ip}\beta_{pq}\wedge\omega^q_i\big)\big] =0
\end{align*}
for all $\xi\in\kf$. It means that the unique allowed momentum level set for this symmetry is $J=0$; accordingly, the isotropy group of this level set is $K$, and
\[
 J^{-1}(0)=W_{\rm PG}.
\]

The other ingredient needed in Routh reduction is the factorization of the metricity bundle~$I^m_{\rm PG}$ induced by a connection $\omega_K$ on the underlying bundle $p_K^{LM}\colon LM\rightarrow\Sigma$, where
\[
 \tau_\Sigma\colon \ \Sigma:=LM/K\to M
\]
is the \emph{bundle of metrics of signature $\eta$}. We will carry out this task in Section~\ref{sec:cont-bundle-decomp}; here we will construct this connection. To this end, consider the decomposition associated to the matrix $\eta$ (see Section~\ref{sec:LFT}). The connection $\omega_K$ on the bundle $p_K^{LM}\colon LM\rightarrow\Sigma$ is induced by this decomposition, namely
\[
 \omega_K:=\pi_\kf\circ\omega_0,
\]
where $\pi_\kf\colon \mathfrak{gl}(m)\to\kf$ is the canonical projector onto the $\kf$-factor in the Cartan decomposition and $\omega_0$ is a connection form on the principal bundle $\tau\colon LM\to M$. The $K$-invariance of the factor $\pf$,
\[
 {\rm Ad}_A\pf\subset\pf\qquad\forall\, A\in K
\]
ensures us that it has the expected properties of a connection.

\subsection{Reduced bundle for Palatini gravity}\label{sec:reduc-bundle-palat}

We have singled out the symmetries of our formulation of Palatini gravity; they are described by the Lie group~$K$. On the other hand, Palatini gravity can be formulated through the variational problem
\[
 \big(\tau_1'\circ\tau_{\rm PG}\colon W_{\rm PG}\to M,\lambda_{\rm PG},0\big).
\]
The relevant bundles in this triple fit in the following diagram
\[
 \begin{tikzcd}[row sep=1.3cm,column sep=1.1cm]
 W_{\rm PG}
 \arrow{dr}{\tau_1'\circ\tau_{\rm PG}}
 \arrow{rr}{\tau_{\rm PG}}
 &
 &
 \cT_0
 \arrow[swap]{dl}{\tau_1'}
 \\
 &
 M
 &
 \end{tikzcd}
\]
It should be stressed that the link between this diagram and the original variational problem depends on the facts that $\cT_0\subset J^1\tau$ and $W_{\rm PG}\subset\wedge^m \big(T^*\big(J^1\tau\big)\big)$: The first inclusion allows us to interpret some fields as derivatives, while the second one tells us that other degrees of freedom behave like (multi)momenta. These correspondences could be lost when symmetries are singled out; by performing the quotients, the corresponding diagram becomes
\[
 \begin{tikzcd}[row sep=1.3cm,column sep=1.1cm]
 W_{\rm PG}/K
 \arrow{dr}{}
 \arrow{rr}{}
 &
 &
 \cT_0/K
 \arrow[swap]{dl}{}
 \\
 &
 M
 &
 \end{tikzcd}
\]
The immediate problem is to find a variational problem associated to this diagram. As we said above, this is difficult to achieve, because neither $\cT_0/K$ is a subset of a jet bundle nor $W_{\rm PG}/K$ is a subbundle of a space of forms. One of the objectives of Routh reduction is to identify in these quotient spaces the degrees of freedom that can be described in this way, and to deal with those that cannot be fitted in this classification. In order to proceed with this identification, the connection $\omega_K$ defined in Section~\ref{sec:momentum-map} is used. Here we will carry out this operation on the quotient space $\cT_0/K$, leaving the discussion of the splitting of $W_{\rm PG}/K$ for later (see Section~\ref{sec:cont-bundle-decomp}). Now, using the adjoint bundle $\tau_{\kf}\colon \widetilde{\kf}\to\Sigma$, the following result holds.

\begin{Proposition} The map
 \begin{align*}
 \Upsilon_{\omega}\colon \ J^1\tau&\longto \big(p_K^{LM}\big)^*\big(J^1\tau_\Sigma\times_\Sigma\Lin\big(\tau_\Sigma^*TM,\widetilde{\kf}\big)\big),\\
 j_x^1s&\longmapsto \big(s(x),j^1_x [s ]_K, [s(x),\omega_K\circ T_xs ]_K\big).
 \end{align*}
is a bundle isomorphism.
\end{Proposition}
The inverse of $\Upsilon_\omega$ is given by
\begin{gather*}%\label{eq:InverseUpsilon}
 \Upsilon_\omega^{-1}\big(e,j_x^1\overline{s},\big[e,\widehat{\xi}\big]_K\big)=\big[v_x\in T_xM\xmapsto{\hspace{.7cm}} (T_x\overline{s} (v_x))_e^H+\big(\widehat{\xi} (v_x)\big)_{LM} (e)\big],
\end{gather*}
where $ (\cdot )^H_e$, $e\in LM,$ is the horizontal lift associated to $\omega_K$.

The map $\Upsilon_\omega$ enjoys a useful property: under this identification, the action of $K$ on $J^1\tau$ is simply
\[
g\cdot\big(e,j_x^1\overline{s},[e,\widehat{\xi}]_K\big)=\big(g\cdot e,j_x^1\overline{s},[e,\widehat{\xi}]_K\big).
\]
This is a direct consequence of the equivariance of the principal connection~$\omega_K$. As a result, we get the following corollary, which is well-known in the literature on Lagrangian reduction \cite{2003CMaPh.236..223L,lopez00:_reduc_princ_fiber_bundl,ellis11:_lagran_poinc}.
\begin{Corollary}\label{cor:identification}
There is an identification
 \[
 J^1\tau/K\simeq J^1\tau_\Sigma\times_{\Sigma}\Lin\big(\tau_\Sigma^*TM,\widetilde{\kf}\big).
 \]
\end{Corollary}

\begin{Remark}
 The choice of a connection on the bundle $p_K^{LM}$ allows us to establish a relationship between the quotient space $J^1\tau/K$ and the jet bundle of the metric bundle~$J^1\tau_\Sigma$, the latter being the relevant bundle in the Einstein--Hilbert approach to relativity, which will be studied in detail in Section~\ref{sec:metr-cont-struct}.
\end{Remark}

Motivated by these considerations, we are in position to split the quotient bundle $J^1\tau/K$ into a jet bundle part and a set of complementary degrees of freedom. It suggests the following definition.
\begin{Definition}[quotient bundle for Palatini gravity]\label{def:reduc-bundle-palat}
 The bundle $J^1\tau_\Sigma\times_\Sigma\mathop{{\rm Lin}}\big(\tau_\Sigma^*TM,\widetilde{\kf}\big)$ is the \emph{quotient bundle for Palatini gravity}.
\end{Definition}
In the next sections, we will explore a further simplification for this bundle as well as a reduction for the Lagrangian responsible of the dynamics on these bundles.

\subsection{Routh reduction scheme for Palatini gravity}\label{sec:routh-reduction}

Our aim is to interpret usual Einstein--Hilbert variational problem as a Routh reduction of the Griffiths variational problem for Palatini gravity. As far as I know, there is no formulation of this type of reduction that could be used in dealing with a Griffiths variational problem, so it is necessary to generalize the techniques employed in~\cite{Capriotti2019} for Routh reduction in field theory to cover this case. The variational problems to be related by this procedure are the Griffiths variational problem for Palatini gravity, described in Definition~\ref{def:VariationalProblemPalatini}, and a variational problem which has not been determined yet, but whose underlying bundle would be related to the quotient bundle given by Definition \ref{def:reduc-bundle-palat}. Now, in Routh reduction the equivalence between the extremals is restricted to those having a particular value of the momentum map;\footnote{In our case this consideration is superfluous, as the momentum map assumes just one value, but it is the way the method proceeds in the general case.} hence, it has some advantages to work in the unified formalism, where momentum level sets have a straightforward meaning (this fact was first recognized in~\cite{GarciaToranoAndres2016291}). Moreover, as the setting of the unified formalism is an affine subbundle of a bundle of forms (see equation~\eqref{eq:LepageEquivProblemGeneral}), the proof of the equivalence between the equations of motion is less involved. A partially filled diagram could clarify these considerations:
\[
 \begin{tikzcd}[row sep=1.3cm,column sep=1.1cm]
 \left(W_{\rm PG},\lambda_{\rm PG},0\right)
 \arrow[leftrightarrow,swap]{d}{\circled{A}}
 \arrow[leftrightarrow,swap]{r}{\circled{1}}
 \arrow[leftrightarrow,to=1-3,bend left=20,"\circled{1}\,+\,\circled{2}"]
 &
 \left(W_{\rm EH}',\lambda_{\rm EH}',0\right)
 \arrow[leftrightarrow]{d}{\circled{B}}
 \arrow[leftrightarrow,swap]{r}{\circled{2}}
 &
 (W_{\rm EH},\lambda_{\rm EH},0 )
 \arrow[leftrightarrow]{d}{\circled{C}}
 \\
 (\cT_0,\cL_{\rm PG}, \langle\omega_\pf \rangle)
 &
\big(J^1\tau_\Sigma\times_\Sigma\mathop{{\rm Lin}} \big(\tau_\Sigma^*TM,\widetilde{\kf}\big) ,\,\bigstar\,,\,\bigstar\,\big)
 \arrow[leftrightarrow]{r}{\circled{3}}
 &
\big(J^1\tau_\Sigma,\,\bigstar\,,\,\bigstar\,\big)
 \end{tikzcd}
\]
Stars ($\bigstar$) refer to geometrical structures (Lagrangians forms and differential constraints) not identified yet. Arrows $\circled{A},\circled{B}$ and $\circled{C}$ correspond to the equivalence given by the Lepage-equivalent construction detailed at the end of Section~\ref{sec:vari-probl-unif}; Theorems~\ref{sec:routh-reduct-palat} and~\ref{thm:Reconstruction} below proved the equivalence indicated by the composition of the horizontal arrows $\circled{1}+\circled{2}$. In order to carry out this operation, we will see in the next sections that the reduced variational problem on $J^1\tau_\Sigma\times_\Sigma\mathop{{\rm Lin}}{\big(\tau_\Sigma^*TM,\widetilde{\kf}\big)}$ can be further simplified to a variational problem on~$J^1\tau_\Sigma$; this fact is depicted in the diagram by arrow $\circled{3}$, and will be carried out in Section~\ref{sec:metr-cont-struct}.

\section{Local coordinates expressions}\label{sec:local-coord-expr}

Here we will obtain some identities allowing us to write down the isomorphism $\Upsilon_\omega^{-1}$ in local terms. In order to proceed, we fix a coordinate chart on~$M$, inducing coordinates $(x^\mu,e_k^\mu)$ on~$LM$. As usual, we will indicate with $ (x^\mu,e_k^\mu,e^\mu_{k\sigma})$ the coordinates induced on $J^1\tau$. As shown above, there exists a set of coordinates $(x^\mu,g^{\mu\nu},\Gamma^\sigma_{\mu\nu})$ on
\[
 J^1\tau/K=\Sigma\times C(LM)
\]
and adapted to this decomposition, namely
\[
 p_K^{LM}\big(x^\mu,e_k^\mu\big)=\big(x^\mu,\eta^{kl}e_k^\mu e^\nu_{l}\big).
\]
In terms of these coordinates, we have
\[
 p^{J^1\tau}_K\big(x^\mu,e_k^\mu,e^\mu_{k\sigma}\big)=\big(x^i,\eta^{ij}e_i^\mu e_j^\nu,-e^k_\mu e_{k\nu}^\sigma\big).
\]
It means in particular that
\[
 Tp_{K}^{LM}\left(\frac{\partial}{\partial x^\mu}\right)=\frac{\partial}{\partial x^\mu}
\]
and
\begin{equation}\label{eq:E_Projection}
 Tp_{K}^{LM}\left(\frac{\partial}{\partial e^\mu_k}\right)=\big(\eta^{kq}e^\rho_q\delta^\sigma_\mu+\eta^{kp}e^\sigma_p\delta_\mu^\rho\big) \frac{\partial}{\partial g^{\sigma\rho}}.
\end{equation}
On the other hand, a principal connection on $LM$ can be written as
\[
 \omega_0=-e^l_\mu\big({\rm d}e^\mu_k-f^\mu_{k\sigma}{\rm d}x^\sigma\big)E^k_l,
\]
where $(f^\mu_{k\sigma})$ is a collection of local functions on~$M$; its Christoffel symbols will be
\[
 \overline{\Gamma}^\sigma_{\rho\mu}=-e_\rho^kf_{k\mu}^\sigma.
\]
Given our definition of the connection $\omega_K$ on the $K$-bundle $p_K^{LM}\colon LM\rightarrow\Sigma$, its components become
\[
[(\omega_0)_\kf]_k^l=-\eta_{kp}\big(\eta^{pq}e_\mu^l-\eta^{lq}e^p_\mu\big) \big({\rm d}e_q^\mu-f^\mu_{q\sigma}{\rm d}x^\sigma\big).
\]
Now we will find the horizontal lift defined by $\omega_K$ for vector fields on $\Sigma$:
\begin{Proposition}\label{prop:Local_Horizontal_Lift}
 The horizontal lift of vector fields on $\Sigma$ associated to the connection $\omega_K$ is locally given by
 \begin{gather*}
 \left(\frac{\partial}{\partial x^\mu}\right)^H =\frac{\partial}{\partial x^\mu}+\frac{1}{2}g_{\beta\rho}e^\rho_k\big(g^{\alpha\sigma}\overline{\Gamma}^\beta_{\alpha\mu} -g^{\alpha\beta}\overline{\Gamma}^\sigma_{\alpha\mu}\big)\frac{\partial}{\partial e^\sigma_k},\\
 \left(\frac{\partial}{\partial g^{\mu\nu}}\right)^H =\frac{1}{4}g_{\rho\beta}e^\beta_k\big(\delta_\mu^\alpha\delta^\rho_\nu+ \delta_\nu^\alpha\delta^\rho_\mu\big)\frac{\partial}{\partial e^\alpha_k}.
 \end{gather*}
\end{Proposition}
\begin{proof} See Appendix \ref{sec:proof-prop-ref}.
\end{proof}

This proposition has the following consequence, that will be important to work with the reduction of the Palatini variational problem.
\begin{Corollary}\label{cor:Coordinates-Lift}
 Let $(x^\mu,g^{\mu\nu},g^{\mu\nu}_\sigma)$ be the induced coordinates on $J^1\tau_\Sigma$. Then
 \begin{gather}
 \left(\frac{\partial}{\partial x^\sigma}+g^{\mu\nu}_\sigma\frac{\partial}{\partial g^{\mu\nu}}\right)^H=\frac{\partial}{\partial x^\sigma}+\frac{1}{2}g_{\beta\rho}e^\rho_k\big[g^{\kappa\beta}_\sigma +\big(g^{\alpha\kappa}\overline{\Gamma}^\beta_{\alpha\sigma}-g^{\alpha\beta} \overline{\Gamma}^\kappa_{\alpha\sigma}\big)\big]\frac{\partial}{\partial e^\kappa_k},\label{eq:HorizontalLift}
 \end{gather}
 where $\left(\cdot\right)^H$ is the horizontal lift in the $K$-principal bundle $p_K^{LM}\colon LM\to\Sigma$ for the connection~$\omega_K$.
\end{Corollary}
\begin{proof} According to Proposition \ref{prop:Local_Horizontal_Lift}, we have that
 \begin{gather*}
 \left(\frac{\partial}{\partial x^\sigma}+g^{\mu\nu}_\sigma\frac{\partial}{\partial g^{\mu\nu}}\right)^H\\
 \qquad{}
 =\frac{\partial}{\partial x^\sigma}+\frac{1}{2}g_{\beta\rho}e^\rho_k\big(g^{\alpha\kappa} \overline{\Gamma}^\beta_{\alpha\sigma}-g^{\alpha\beta} \overline{\Gamma}^\kappa_{\alpha\sigma}\big)\frac{\partial}{\partial e^\kappa_k}+\frac{1}{4}g^{\mu\nu}_\sigma g_{\rho\beta}e^\rho_k\big(\delta_\mu^\alpha\delta^\beta_\nu+\delta_\nu^\alpha\delta^\beta_\mu\big)\frac{\partial}{\partial e^\alpha_k}\\
 \qquad{} =\frac{\partial}{\partial x^\sigma}+\frac{1}{2}g_{\beta\rho}e^\rho_k\big[g^{\kappa\beta}_\sigma +\big(g^{\alpha\kappa}\overline{\Gamma}^\beta_{\alpha\sigma}-g^{\alpha\beta} \overline{\Gamma}^\kappa_{\alpha\sigma}\big)\big]\frac{\partial}{\partial e^\kappa_k},
 \end{gather*}
 as required.
\end{proof}

Let us now introduce coordinates on the vector bundle $\widetilde{\kf}$. In order to do this, let us suppose that $\left(\phi=\left(x^\mu\right),U\right)$ is a coordinate chart on $M$; then it is also a trivializing domain for the principal bundle $LM$, where
\[
 t_U\colon \ \tau^{-1}(U)\to U\times {\rm GL}(m)\colon \ u= (X_1,\dots,X_m) \mapsto\big(x^{\mu} (\tau(u)),e_k^\mu(u)\big)
\]
if and only if
\[
 X_k=e_k^\mu(u)\frac{\partial}{\partial x^\mu}.
\]
Hence, we can define the coordinate chart $\big(\phi_{\kf},\tau_{\kf}^{-1}(U)\big)$ \cite{springerlink:10.1007/PL00004852}. In order to proceed, we use the correspondence between the space of sections of the adjoint bundle $\Gamma\tau_\kf$ and the set of $p_{K}^{LM}$-vertical $K$-invariant vector fields on~$LM$.

Therefore, taking the base $\big\{E_\sigma^\rho\big\}$ on $\mathfrak{gl}(m)$ such that
\[
\big(E_\sigma^\rho\big)_\alpha^\beta=\delta^\beta_\sigma\delta_\alpha^\rho,
\]
we can define the set of ${\rm GL}(m)$-invariant $\tau$-vertical vector fields $\widetilde{E}_\sigma^\rho$ whose flow $\Phi^{\widetilde{E}_\sigma^\rho}_t\colon \tau^{-1}(U)\to\tau^{-1}(U)$ is given by
\[
 \Phi^{\widetilde{E}_\sigma^\rho}_t(u):= t_U^{-1}\big(\tau(u),\big[\exp{\big(tE_\sigma^\rho\big)}\big]^\alpha_\beta e_i^\beta(u)\big);
\]
it means that, locally, these vector fields are such that
\begin{gather}\label{eq:InvarianVectorExpression}
 T_{u}t_U\big(\widetilde{E}_\sigma^\rho(u)\big)=e^\rho_i\frac{\partial}{\partial e_i^\sigma}.
\end{gather}
In the following, we will adopt the usual convention according to which the map $Tt_U$ is not explicitly written, namely, where
\[
 \frac{\partial}{\partial e_i^\sigma}\qquad\text{and}\qquad Tt_U^{-1}\left(\frac{\partial}{\partial e_i^\sigma}\right)
\]
are identified. We can write down any $p_K^{LM}$-vertical $K$-invariant vector field $Z$ on $LM$ as
\[
 Z=A^\rho_\sigma\widetilde{E}^\sigma_\rho;
\]
then, using equation~\eqref{eq:E_Projection}, we obtain the following result.

\begin{Lemma} The vector field on $\tau^{-1}(U)$ given by
 \[
 Z=A^\rho_\sigma\widetilde{E}^\sigma_\rho
 \]
 is $p_K^{LM}$-vertical if and only if
 \[
 g^{\sigma\alpha}A_\alpha^\rho+g^{\rho\alpha}A_\alpha^\sigma=0.
 \]
\end{Lemma}
\begin{proof} In fact, we have that
 \begin{align*}
 0&=T_up_K^{LM}\big(A_\sigma^\rho\widetilde{E}_\rho^\sigma(u)\big)
 =A_\sigma^\rho T_up_K^{LM}\big(\widetilde{E}_\rho^\sigma(u)\big)
 =A_\sigma^\rho e_i^\sigma(u) T_up_K^{LM}\left(\frac{\partial}{\partial e_i^\rho}\right)\\
 &=A_\sigma^\rho e_i^\sigma(u)\eta^{iq}\big[e^\alpha_q(u)\delta^\beta_\rho+e^\beta_q(u)\delta^\alpha_\rho\big]\frac{\partial}{\partial g^{\alpha\beta}},
 \end{align*}
 and the identity follows.
\end{proof}

Therefore, we will have that
\[
 \phi_{\kf}([u,B]_K):=\big(x^\mu(u),\eta^{kl}e_k^\mu(u)e_l^\nu(u),A_{\sigma}^\rho([u,B]_K)\big)
\]
if and only if $g^{\sigma\alpha}A_\alpha^\rho+g^{\rho\alpha}A_\alpha^\sigma=0$ and
\[
[u,B]_K=A_{\sigma}^\rho\widetilde{E}_\rho^\sigma(u).
\]

In order to relate the coordinates $A_\sigma^\rho$ with the element $[u,B]_K$, we need to look closely to the identification between $\Gamma\widetilde{\kf}$ and the set of $p_{K}^{LM}$-vertical $K$-invariant vector fields on~$LM$. It uses the correspondence
\[
 V\tau\simeq LM\times\mathfrak{gl}(m)
\]
given by
\[
(u,B)\mapsto \frac{\vec{{\rm d}}}{{\rm d}t}\bigg|_{t=0}(u\cdot\exp{(-tB)}).
\]
In coordinates it reads
\[
 \big(u=(X_1,\dots,X_m),B=\big(B_i^j\big)\big)\mapsto-B_j^ie_i^\rho\frac{\partial}{\partial e_j^\rho},
\]
and using equation~\eqref{eq:InvarianVectorExpression} it becomes
\[
\big(u= (X_1,\dots,X_m),B=\big(B_i^j\big)\big)\mapsto-B_j^ie_i^\rho e^j_\sigma\widetilde{E}^\sigma_\rho.
\]
Therefore, it results that
\[
 \widehat{A}_\rho^\sigma (u,B)=-e_\rho^iB_i^je_j^\sigma
\]
is a ${\rm GL}(m)$-invariant function on $LM\times\mathfrak{gl}(m)$ when ${\rm GL}(m)$ acts on $\mathfrak{gl}(m)$ by the adjoint action; then, it gives us the set of functions $A_\rho^\sigma$ on $\tau_\kf^{-1}(U)\subset\widetilde{\kf}$ that completes the coordinates $\phi_\kf$.

\begin{Lemma}\label{lem:Coords-On-KF}
 The map $\phi_\kf\colon \tau_\kf^{-1}(U)\to U\times\mathbb{R}^{m^2+\frac{m(m+1)}{2}}$ given by
 \[
 \phi_\kf([u,B]_K)=\big(x^\mu(u),\eta^{kl}e_k^\mu(u)e_l^\nu(u),-e_\rho^i(u)B_i^je_j^\sigma(u)\big)
 \]
 defines a set of coordinates on $\tau_\kf^{-1}(U)$.
\end{Lemma}
\begin{proof}According to the previous discussion, it is only necessary to prove that for any $B\in\kf$, i.e., such that
 \[
 \eta^{ik}B_k^j+\eta^{jk}B_k^i=0,
 \]
 the corresponding element on $T_uLM$,
 \[
 Z=-e_\rho^i(u)B_i^je_j^\sigma(u)\widetilde{E}^\rho_\sigma
 \]
 verifies the constraint
 \[
 T_up_K^{LM} (Z)=0.
 \]
 But it follows that
 \begin{align*}
 g^{\rho\alpha}A_\alpha^\sigma+g^{\sigma\alpha}A_\alpha^\rho& =-g^{\rho\alpha}e_\alpha^iB_i^je_j^\sigma-g^{\sigma\alpha}e_\alpha^iB_i^je_j^\rho
 =-\eta^{ik}B_i^j\big(e_k^\sigma e_j^\rho+e_k^\rho e_j^\sigma\big)\\
 &=-\big(\eta^{ik}B_k^j+\eta^{jk}B_k^i\big)e_j^\rho e_k^\sigma =0,
 \end{align*}
 as required.
\end{proof}

Then, let us point out that Lemma~\ref{lem:Coords-On-KF} allows us to set coordinates on the bundle
\[
 \overline{p}\colon \ \mathop{{\rm Lin}}{\big(\tau_\Sigma^*TM,\widetilde{\kf}\big)}\to\Sigma.
\]
In fact, any element $\left(g_x,\alpha\right)\in\mathop{{\rm Lin}}{\big(\tau_\Sigma^*TM,\widetilde{\kf}\big)}$ admits coordinates $\big(x^\mu,g^{\mu\nu},A_{\sigma\rho}^\mu\big)$ if and only if $\big(x^\mu,g^{\mu\nu}\big)$ are the corresponding coordinates for $g_x\in\Sigma$ and
\[
 \alpha\left(\frac{\partial}{\partial x^\rho}\right)=A_{\sigma\rho}^\mu\widetilde{E}^\sigma_\mu (e_x),
\]
where $e_x\in LM$ is any element in $\big(p_K^{LM}\big)^{-1} (g_x)$.

\section[Torsion, metricity and contact structures on the quotient space]{Torsion, metricity and contact structures\\ on the quotient space}\label{sec:metr-cont-struct}

There are two tasks to carry out in order to understand the Routh reduction of Palatini gravity: We need to characterize the effects produced by the fact that we are working on the subbundle $\cT_0\subset J^1\tau$ instead of the full jet space~$J^1\tau$; additionally, we want to find the differential constraints for the reduced system. Accordingly, in this section we will prove that
\begin{itemize}\itemsep=0pt
\item When restricted to $\cT_0$, the quotient map
 \[
 J^1\tau\to J^1\tau_\Sigma\times_\Sigma\mathop{{\rm Lin}}{\big(\tau_\Sigma^*TM,\widetilde{\kf}\big)}
 \]
 will reduce to
 \[
 \cT_0\to J^1\tau_\Sigma,
 \]
 that is, we can forget about the ``vertical'' degrees of freedom related to the factor \linebreak $\mathop{{\rm Lin}}{\big(\tau_\Sigma^*TM,\widetilde{\kf}\big)}$. This is achieved in Propositions~\ref{prop:RestrictionToI0} and~\ref{prop:RestrictionLeviCivita}, and in Corollary~\ref{cor:IsomorphismI0PrimePrime}.
\item The metricity conditions are the pullback along the quotient map of the contact forms on~$J^1\tau_\Sigma$. This result is very interesting because it tells us that the reduction scheme implemented relates a Griffiths variational problem (Palatini gravity, see Definition~\ref{def:VariationalProblemPalatini}) with a classical variational problem (Einstein--Hilbert gravity as described in Section~\ref{sec:first-order-vari} below). This is accomplished in Proposition~\ref{prop:Metricity-Horizontal}.
\end{itemize}

So, let us use the following diagram
\begin{equation*}%\label{eq:GOmegaDiagram}
 \begin{tikzcd}[row sep=1.3cm,column sep=3.1cm]
 J^1\tau \arrow{r}{\Upsilon_\omega}
 \arrow{d}{p_K^{J^1\tau}}
 &
 \big(p_K^{LM}\big)^*\big(J^1\tau_\Sigma\times_\Sigma\mathop{{\rm Lin}} {\big(\tau_\Sigma^*TM,\widetilde{\kf}\big)}\big)
 \arrow{d}{{\rm pr}_{23}}
 \\
 J^1\tau/K
 \arrow{r}{\overline{\Upsilon}_\omega}
 \arrow{d}{\sim}
 &
 J^1\tau_\Sigma\times_\Sigma\mathop{{\rm Lin}}{\big(\tau_\Sigma^*TM,\widetilde{\kf}\big)}
 \arrow{dl}{G_\omega}
 \\
 \Sigma\times C(LM)
 &
 \end{tikzcd}
\end{equation*}
in order to define the diffeomorphism $G_\omega$; here $\overline{\Upsilon}_\omega$ is the map induced by $\Upsilon_\omega$. Therefore, let us construct the pullback bundle ${\rm pr}_1\colon \big(p_K^{LM}\big)^*\big(J^1\tau_\Sigma\big)\to LM$ by means of the commutative diagram
\[
 \begin{tikzcd}[row sep=1.3cm,column sep=3.1cm]
 \big(p_K^{LM}\big)^*\big(J^1\tau_\Sigma\big)
 \arrow[swap]{r}{{\rm pr}_2}
 \arrow{d}{{\rm pr}_1}
 &
 J^1\tau_\Sigma
 \arrow{d}{(\tau_\Sigma)_{10}}
 \\
 LM
 \arrow{r}{p_K^{LM}}
 &
 \Sigma
 \end{tikzcd}
\]
In this setting, we can prove that the zero torsion submanifold~$\cT_0$ has some nice properties regarding the decomposition induced by the connection~$\omega_K$.
\begin{Proposition}\label{prop:RestrictionToI0}
 The canonical projection
 \begin{align*}
 {\rm pr}_\Sigma\colon \ \big(p_K^{LM}\big)^*&\big(J^1\tau_\Sigma\times_\Sigma\Lin \big(\tau_\Sigma^*TM,\widetilde{\kf}\big)\big)\longto\big(p_K^{LM}\big)^*\big(J^1\tau_\Sigma\big)\\
 &\big(e,j_x^1\overline{s},\big[e,\widehat{\xi}\big]_K\big)\xmapsto{\hspace{2.2cm}} \big(e,j_x^1\overline{s}\big)
 \end{align*}
 restricted to the submanifold
 \[
 \cT_0':=\Upsilon_\omega (\cT_0 )\subset\big(p_K^{LM}\big)^*\big(J^1\tau_\Sigma\times_\Sigma\Lin \big(\tau_\Sigma^*TM,\widetilde{\kf}\big)\big)
 \]
 is a diffeomorphism between $\cT_0'$ and $\big(p_K^{LM}\big)^*\big(J^1\tau_\Sigma\big)$.
\end{Proposition}
\begin{proof} The proof of this proposition will be local.
 Using equation~\eqref{eq:HorizontalLift} and the coordinates introduced above, we have that
 \begin{align*}
 \frac{\partial}{\partial x^\sigma}+e^{\mu}_{k\sigma}\frac{\partial}{\partial e^\mu_k}&=
 \left(\frac{\partial}{\partial x^\sigma}+g^{\mu\nu}_\sigma\frac{\partial}{\partial g^{\mu\nu}}\right)^H+A_{\rho\sigma}^\mu\widetilde{E}^\rho_\mu (e_x )\\
 &=\frac{\partial}{\partial x^\sigma}+\frac{1}{2}g_{\beta\rho}e^\rho_k\big[g^{\mu\beta}_\sigma+ \big(g^{\alpha\mu}\overline{\Gamma}^\beta_{\alpha\sigma}-g^{\alpha\beta} \overline{\Gamma}^\mu_{\alpha\sigma}\big)\big]\frac{\partial}{\partial e^\mu_k}-A_{\rho\sigma}^\mu e^\rho_k\frac{\partial}{\partial e^\mu_k},
 \end{align*}
 namely
 \[
 e^{\mu}_{k\sigma}=\frac{1}{2}g_{\beta\rho}e^\rho_k\big[g^{\mu\beta}_\sigma +\big(g^{\alpha\mu}\overline{\Gamma}^\beta_{\alpha\sigma}-g^{\alpha\beta} \overline{\Gamma}^\mu_{\alpha\sigma}\big)\big]-A_{\rho\sigma}^\mu e^\rho_k.
 \]
 Then it follows that, for the $K$-invariant functions $\Gamma^\mu_{\nu\sigma}$,
 \begin{gather}
 \Gamma_{\rho\sigma}^\mu =-e^k_\rho e^{\mu}_{k\sigma}
 =-\frac{1}{2}g_{\beta\rho}\big[g^{\mu\beta}_\sigma +\big(g^{\alpha\mu}\overline{\Gamma}^\beta_{\alpha\sigma} -g^{\alpha\beta}\overline{\Gamma}^\mu_{\alpha\sigma}\big)\big]+A_{\rho\sigma}^\mu. \label{eq:GammaInTermsQuotient}
 \end{gather}
 It means that the set $\cT_0'$ is locally given by the equation
 \begin{gather*}
 \frac{1}{2}g_{\beta\sigma}\big[g^{\mu\beta}_\rho+\big(g^{\alpha\mu} \overline{\Gamma}^\beta_{\alpha\rho}-g^{\alpha\beta}\overline{\Gamma}^\mu_{\alpha\rho}\big)\big] -\frac{1}{2}g_{\beta\rho}\big[g^{\mu\beta}_\sigma+\big(g^{\alpha\mu} \overline{\Gamma}^\beta_{\alpha\sigma}-g^{\alpha\beta}\overline{\Gamma}^\mu_{\alpha\sigma}\big)\big]
 +A_{\rho\sigma}^\mu-A_{\sigma\rho}^\mu=0.
 \end{gather*}
 Let us define the set of quantities
 \[
 A_{\mu\nu\sigma}:=g_{\mu\rho}A^\rho_{\nu\sigma};
 \]
 then using this equation and the fact that
 \[
 A_{\mu\nu\sigma}+A_{\nu\mu\sigma}=g_{\mu\rho}A^\rho_{\nu\sigma}+g_{\nu\rho}A^\rho_{\mu\sigma}=0,
 \]
 we can conclude, from Proposition \ref{Prop:UniqueSolution}, that the elements $A^\mu_{\nu\sigma}$ are uniquely determined by the fact that they belong to $\cT_0'$. In other words, the set
 \[
({\rm pr}_\Sigma)^{-1}\big(e,j_x^1\overline{s}\big)\cap\cT_0'
 \]
 consists in a single element.
\end{proof}

Proposition \ref{prop:RestrictionToI0} can be geometrically interpreted: Recall that, through isomorphism
\[
 G_\omega\colon \ J^1\tau_\Sigma\times_\Sigma\mathop{{\rm Lin}}{\big(\tau_\Sigma^*TM,\widetilde{\kf}\big)}\to\Sigma\times C(LM),
\]
any section of the reduced bundle $J^1\tau_\Sigma\times_\Sigma\mathop{{\rm Lin}}{\big(\tau_\Sigma^*TM,\widetilde{\kf}\big)}$ can be seen as a pair ``metric'' plus ``connection''. With this interpretation in mind, the previous proposition tells us that, when projected to the quotient, the ``connection part'' of the section corresponds to Levi-Civita connection, and so, it is uniquely determined by its ``metric part''. The following result summarizes it.

\begin{Proposition}\label{prop:RestrictionLeviCivita}
 Let
 \[
 \sigma\colon \ M\rightarrow J^1\tau_\Sigma\times_\Sigma\mathop{{\rm Lin}}{\big(\tau_\Sigma^*TM,\widetilde{\kf}\big)}
 \]
 be a section of the composite map
 \[
 J^1\tau_\Sigma\times_\Sigma\mathop{{\rm Lin}}{\big(\tau_\Sigma^*TM,\widetilde{\kf}\big)}{\longto} \Sigma\stackrel{\tau_\Sigma}{\longto}M
 \]
 such that ${\rm pr}_1\circ\sigma\colon M\to J^1\tau_\Sigma$ is a holonomic section and
 \[
 \mathop{{\rm Im}}{\sigma}\subset{\rm pr}_{23} (\cT_0' ).
 \]
 Then
 \[
 \Gamma_\sigma:={\rm pr}_2\circ G_\omega\circ\sigma\colon \ M\rightarrow C_0(LM)
 \]
 is the Levi-Civita connection associated to the metric $g_\sigma:={\rm pr}_1\circ G_\omega\circ\sigma$.
\end{Proposition}
\begin{proof}Locally, the map $G_\omega$ is given by equation~\eqref{eq:GammaInTermsQuotient}. Therefore, from the proof of the previous Proposition and using Proposition~\ref{Prop:UniqueSolution}, we will have that the elements
 \[
 \Gamma_{\mu\nu\sigma}:=g_{\mu\rho}\Gamma^\rho_{\nu\sigma}
 \]
 are uniquely determined by the set of equations
 \begin{gather}
 \Gamma_{\mu\rho\sigma}-\Gamma_{\rho\mu\sigma}=0,\qquad
 \Gamma_{\mu\rho\sigma}+\Gamma_{\rho\mu\sigma} =-g_{\mu\alpha}g_{\rho\beta}g^{\alpha\beta}_\sigma.\label{eq:GInTermsGamma}
 \end{gather}
 It means that
 \[
 \Gamma_{\mu\rho\sigma}=-\frac{1}{2}\big(g_{\mu\alpha}g_{\rho\beta}g^{\alpha\beta}_\sigma +g_{\rho\alpha}g_{\sigma\beta}g^{\alpha\beta}_\mu-g_{\sigma\alpha}g_{\mu\beta}g^{\alpha\beta}_\mu\big).
 \]
 Now, using the definition
 \[
 g_{\mu\nu,\sigma}:=-g_{\mu\alpha}g_{\nu\beta}g^{\alpha\beta}_\sigma
 \]
 we obtain
 \begin{gather}\label{eq:GammaInTermsQuotientJetVariables}
 \Gamma_{\rho\sigma}^\mu=\frac{1}{2}g^{\mu\alpha}\big(g_{\alpha\nu,\sigma} +g_{\rho\sigma,\alpha}-g_{\sigma\alpha,\nu}\big).
 \end{gather}
 Because ${\rm pr}_1\circ\sigma$ is holonomic, we have that
 \[
 g^{\mu\nu}_\sigma=\frac{\partial g^{\mu\nu}}{\partial x^\sigma},
 \]
 as required.
\end{proof}

Let us define
\[
 \cT_0'':=G_\omega^{-1} (\Sigma\times C_0(LM) )=\overline{\Upsilon}_\omega\big(p_K^{J^1\tau}(\cT_0)\big);
\]
then, we need to draw our attention to the diagram in Fig.~\ref{fig:MapsInvolvedRouth}.
\begin{figure}[h!]
\scalebox{.85}
{
 $$
 \begin{tikzcd}[row sep=1cm,column sep=.1cm,ampersand replacement=\&]
 \cT_0
 \arrow{rrr}{\Upsilon_\omega|_{\cT_0}}
 \arrow[hookrightarrow]{rd}{}
 \arrow[swap]{d}{p_K^{J^1\tau}\big|_{\cT_0}}
 \&
 \&
 \&
 \cT_0'
 \arrow[hookrightarrow]{dl}{}
 \arrow{dd}{{\rm pr}_\Sigma|_{\cT_0'}}
 \\
 \cT_0/K
 \arrow[swap]{d}{\sim}
 \arrow{ddr}{\overline{\Upsilon}_\omega|_{\cT_0/K}}
 \&
 J^1\tau
 \arrow{r}{\Upsilon_\omega}
 \arrow{dd}{\overline{\Upsilon}_\omega\circ p_K^{J^1\tau}}
 \&
\big(p_k^{LM}\big)^*\big(J^1\tau_\Sigma\times_\Sigma\mathop{{\rm Lin}}{\big(\tau_\Sigma^*TM,\widetilde{\kf}\big)}\big)
 \arrow{ddl}{{\rm pr}_{23}}
 \arrow{dr}{{\rm pr}_\Sigma}
 \&
 \\
 \Sigma\times C_0(LM)
 \arrow[swap]{dd}{G_\omega^{-1}\big|_{\Sigma\times C_0(LM)}}
 \&
 \&
 \&
\big(p_k^{LM}\big)^*\big(J^1\tau_\Sigma\big)
 \arrow{dd}{{\rm pr}_2}
 \\
 \&
 J^1\tau_\Sigma\times_\Sigma\mathop{{\rm Lin}}{\big(\tau_\Sigma^*TM,\widetilde{\kf}\big)}
 \arrow{drr}{{\rm pr}_1}
 \&
 \&
 \\
 \cT_0''
 \arrow[hookrightarrow]{ru}{}
 \arrow{rrr}{{\rm pr}_1|_{\cT_0''}}
 \&
 \&
 \&
 J^1\tau_\Sigma
 \end{tikzcd}
 $$}
\caption{Maps involved in the Routh reduction of Palatini gravity.}\label{fig:MapsInvolvedRouth}
\end{figure}
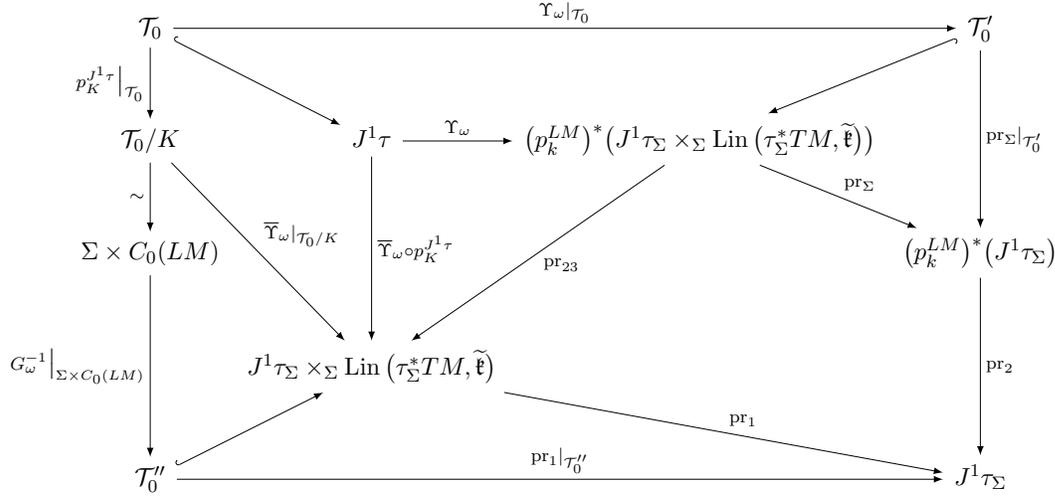

As a consequence of formula~\eqref{eq:GammaInTermsQuotientJetVariables}, we obtain the following corollary; in short, it says that in the reduced bundle $J^1\tau_\Sigma\times_\Sigma\mathop{{\rm Lin}}{\big(\tau_\Sigma^*TM,\widetilde{\kf}\big)}$, the degrees of freedom associated to the factor $\mathop{{\rm Lin}}{\big(\tau_\Sigma^*TM,\widetilde{\kf}\big)}$ are superfluous.
\begin{Corollary}\label{cor:IsomorphismI0PrimePrime}
 The map
 \[
{\rm pr}_1|_{\cT_0''}\colon \ \cT_0''\to J^1\tau_\Sigma
 \]
 is a bundle isomorphism over the identity on $\Sigma$.
\end{Corollary}
\begin{proof}
 Locally, composite map $\overline{\Upsilon}_\omega\circ p_K^{J^1\tau}$ is given by
 \[
 \overline{\Upsilon}_\omega\circ p_K^{J^1\tau}\big(\big[x^\mu,e_k^\nu,e^\sigma_{k\rho}\big]_K\big)=\big(x^\mu,\eta^{ij}e_i^\mu e_j^\nu,g^{\mu\nu}_\sigma\big),
 \]
 where coordinates $g_\sigma^{\mu\nu}$ are calculated using equation~\eqref{eq:GInTermsGamma}.
\end{proof}

For the last result of the section, we need any of the composite maps
\[
 \begin{tikzcd}[row sep=1.3cm,column sep=.7cm]
 J^1\tau
 \arrow{r}{\Upsilon_\omega}
 \arrow{d}{\overline{\Upsilon}_\omega\circ p^{J^1\tau}_K}
 &
\big(p_{K}^{LM}\big)^*\big(J^1\tau_\Sigma\times_\Sigma\mathop{{\rm Lin}}{\big(\tau_\Sigma^*TM,\widetilde{\kf}\big)}\big)
 \arrow{r}{{\rm pr}_\Sigma}
 &
 \big(p_K^{LM}\big)^*\big(J^1\tau_\Sigma\big)
 \arrow{d}{{\rm pr}_2}
 \\
 J^1\tau_\Sigma\times_\Sigma\mathop{{\rm Lin}}{\big(\tau_\Sigma^*TM,\widetilde{\kf}\big)}
 \arrow{rr}{{\rm pr}_1}
 &
 &
 J^1\tau_\Sigma
 \end{tikzcd}
\]

So far, we have obtained a result allowing us to reduce the bundle $J^1\tau_\Sigma\times_\Sigma\mathop{{\rm Lin}}{\big(\tau_\Sigma^*TM,\widetilde{\kf}\big)}$ further down to $J^1\tau_\Sigma$; thus, we are halfway to connect the reduced variational problem defined on this bundle with the Einstein--Hilbert variational problem. Now, as we mentioned above, the splitting induced by the connection form $\omega_K$ allows us to relate the metricity forms with a contact structure on the quotient bundle; it will make possible to complete this connection by showing that the contact structure on $J^1\tau_\Sigma$ is a sort of reduction structure for metricity conditions.
\begin{Proposition}\label{prop:Metricity-Horizontal}
 The metricity forms are $({\rm pr}_2\circ{\rm pr}_\Sigma\circ\Upsilon_\omega)$-horizontal $($also $\big({\rm pr}_1\circ\overline{\Upsilon}_\omega\circ p^{J^1\tau}_K\big)$-horizontal$)$. In fact,
 \[
 Tp_K^{LM}\circ\omega_\pf= ({\rm pr}_2\circ{\rm pr}_\Sigma\circ\Upsilon_\omega )^*\overline{\omega}=\big({\rm pr}_1\circ\overline{\Upsilon}_\omega\circ p^{J^1\tau}_K\big)^*\overline{\omega}
 \]
 where $\overline{\omega}$ is the contact form on $J^1\tau_\Sigma$.
\end{Proposition}
\begin{proof}
 In local coordinates, we have that
 \[
 ({\rm pr}_2\circ{\rm pr}_\Sigma\circ\Upsilon_\omega )\big(x^\mu,e^\nu_k,e^\nu_{k\sigma}\big)=\big(x^\mu,g^{\mu\nu},g^{\mu\nu}_{\sigma}\big),
 \]
 where $g^{\mu\nu}_{\sigma}$ is calculated using equation~\eqref{eq:GammaInTermsQuotient}. On the other hand, the metricity forms have the following local expression \cite{capriotti14:_differ_palat}
 \begin{equation}\label{eq:GammaInLocalTerms}
 \eta^{ik}\omega_k^j+\eta^{jk}\omega_k^i=e^i_\mu e^j_\nu\big[{\rm d}g^{\mu\nu}+\big(g^{\mu\sigma}\Gamma_{\sigma\rho}^\nu+g^{\nu\sigma} \Gamma_{\sigma\rho}^\mu\big){\rm d}x^\rho\big].
 \end{equation}
 Using equation~\eqref{eq:GammaInTermsQuotient}, it follows that
 \[
 \eta^{ik}\omega_k^j+\eta^{jk}\omega_k^i=e^i_\mu e^j_\nu\big({\rm d}g^{\mu\nu}-g^{\mu\nu}_\sigma {\rm d}x^\sigma\big),
 \]
 namely, the metricity condition is horizontal with respect to the projection
 \[
 {\rm pr}_2\circ{\rm pr}_\Sigma\circ\Upsilon_\omega\colon \ J^1\tau\longto J^1\tau_\Sigma,
 \]
 and the form in the base manifold is nothing, but the generator of the contact structure.
\end{proof}

\section{First order variational problem for Einstein--Hilbert gravity}\label{sec:first-order-vari}

Einstein--Hilbert variational problem is a classical second order variational problem on the bundle of metrics $\tau_\Sigma\colon \Sigma\to M$ of signature~$\eta$ on~$M$~-- see \cite{doi:10.1063/1.4890555, doi:10.1063/1.4998526} and references therein. It means that its dynamics is dictated by a Lagrangian density
\[
 \cL_{\rm EH}\colon J^2\tau_\Sigma\to\wedge^m (T^*M ),
\]
given essentially by the scalar curvature of the Levi-Civita connection associated to a metric. In terms of our nomenclature regarding variational problems, this variational problem is prescribed by the triple{\samepage
\[
 \big((\tau_{\Sigma})_2\colon J^2\tau_\Sigma\to M,\cL_{\rm EH},\mathcal{I}_{{\rm con}}^\Sigma\big),
\]
with $\mathcal{I}_{{\rm con}}^\Sigma$ referring to the exterior differential system on $J^2\tau_\Sigma$ associated to its contact structure.}

The main objective of this section is to provide a first order variational problem for Einstein--Hilbert gravity. Namely, we have proved in Section~\ref{sec:metr-cont-struct} that the quotient bundle given by Definition~\ref{def:reduc-bundle-palat} can be further reduced to $J^1\tau_\Sigma$; also, it was shown (see Proposition~\ref{prop:Metricity-Horizontal}) that contact structure on this jet space is related to metricity conditions via reduction map. Therefore, we will construct a variational problem on this bundle $J^1\tau_\Sigma$ and we will prove that this variational problem describes Einstein--Hilbert gravity. Later on, we will prove that this variational problem can be interpreted as Routh reduction of the variational problem for Palatini gravity as defined in Section~\ref{def:VariationalProblemPalatini}.
\begin{Definition}[Einstein--Hilbert Lagrangian form]
The \emph{Einstein--Hilbert Lagrangian form} is the unique $2$-horizontal $m$-form $\cL_{\rm EH}^{(1)}$ on $J^1\tau_\Sigma$ such that
\[
\big({\rm pr}_1\circ\overline{\Upsilon}_\omega\circ p_K^{J^1\tau}\big)^*\cL_{\rm EH}^{(1)}=i_0^*\cL_{\rm PG}.
\]
\end{Definition}
Recall also that in local terms, Palatini Lagrangian~\eqref{eq:PalatiniLagrangianInvariant} can be written as
\begin{gather}\label{eq:QuotientPalatiniLag}
 {\cL_{\rm PG}}=\epsilon_{\mu_1\cdots\mu_{n-2}\gamma\kappa}\sqrt{|\det{g}|}g^{\kappa\phi} {\rm d} x^{\mu_1}\wedge\!\cdots\!\wedge {\rm d} x^{\mu_{n-2}}\!\wedge\!\big({\rm d}\Gamma^\gamma_{\rho\phi}\wedge {\rm d} x^\rho+\Gamma^\sigma_{\delta\phi}\Gamma^\gamma_{\beta\sigma}{\rm d} x^\beta\wedge {\rm d} x^\delta\big).\!\!\!
\end{gather}
Using equation~\eqref{eq:GammaInTermsQuotientJetVariables}
we see that $\cL_{\rm EH}^{(1)}$ has the same form than $\cL_{\rm PG}$, but replacing $\Gamma_{\mu\nu}^\sigma$ by their expressions in terms of the jet variables $g_{\sigma}^{\mu\nu}$. Nevertheless, it is not yet Einstein--Hilbert Lagrangian density because is it neither a density nor a functional on $J^2\tau_\Sigma$.

We are pursuing here to establish the equivalence between the classical variational problem associated to the Lagrangian density $\cL_{\rm EH}\colon J^2{\tau}_\Sigma\rightarrow\wedge^m\left(T^*M\right)$ and the variational problem $\big(J^1{\tau}_\Sigma,\cL_{\rm EH}^{(1)},\mathcal{I}^\Sigma_{{\rm con}}\big)$. As we have said above, the main difference between these variational problems is related to the nature of the Lagrangian form. In the latter, this form is not a horizontal form on $J^1{\tau}_\Sigma$, whereas in the former case, the Lagrangian form on~$J^2{\tau}_\Sigma$ is specified through a~Lagrangian density, giving rise to a horizontal form on this jet bundle. The following lemma tells us how these Lagrangians are related. In order to formulate this result, the definition of the \emph{horizontalization operator} $h\colon \Omega^m\left(J^1{\tau}_\Sigma\right)\rightarrow\Omega^m\big(J^2{\tau}_\Sigma\big)$ should be kept in mind \cite{krupka15:_introd_global_variat_geomet_atlan}; in fact, defining the map
 \[
 h_{j_x^2s}:=T_xj^1s\circ T_{j^2_xs} ({\tau}_\Sigma )_2\colon \ T_{j^2_xs}\big(J^2\tau_\Sigma\big)\to T_{j^1_xs}\big(J^1\tau_\Sigma\big),
 \]
 where $(\tau_\Sigma)_2\colon J^2\tau_\Sigma\to M$ is the canonical projection of the $2$-jet bundle of the metric bundle onto $M$, we have
 \[
 h (\alpha)\big|_{j_x^2s}:= \alpha\big|_{j_x^1s}\circ h_{j_x^2s}
 \]
 for every $\alpha\in\Omega^m\big(J^1\tau_\Sigma\big)$ and $j_x^2s\in J^2\tau_\Sigma$.
\begin{Lemma}\label{lem:DensityAndLagFormEH}
 It results that
 \[
 \cL_{\rm EH}=h\big(\cL_{\rm EH}^{(1)}\big)
 \]
 for $h\colon \Omega^m\big(J^1{\tau}_\Sigma\big)\rightarrow\Omega^m\left(J^2{\tau}_\Sigma\right)$ the horizontalization operator.
\end{Lemma}
\begin{proof}
 In terms of the coordinates $\big(x^\mu,g^{\mu\nu},g^{\mu\nu}_{\alpha},g^{\mu\nu}_{\alpha\beta}\big)$ on $J^2{\tau}_\Sigma$, we have that
 \[
 h\big({\rm d}g^{\mu\nu}\big)=g^{\mu\nu}_\alpha {\rm d}x^\alpha,\qquad h\big({\rm d}g^{\mu\nu}_\alpha\big)=g^{\mu\nu}_{\alpha\beta}{\rm d}x^\beta.
 \]
 The result follows from a (rather lenghty) calculation, using expression~\eqref{eq:QuotientPalatiniLag} and the formula for the Christoffel symbols~\eqref{eq:GammaInTermsQuotientJetVariables}.
\end{proof}

The occurrence of the horizontalization operator in this lemma is crucial for our purposes, as the following proposition shows.

\begin{Theorem}\label{thm:IntegralsAndHorizontalization}
 Let $\pi\colon E\rightarrow M$ be a bundle on a $($compact$)$ manifold $M$ of dimension $m$. For any $\alpha\in\Omega^m\big(J^k\pi\big)$ and any section $s\in\Gamma\pi$, we have that
 \[
 \int_M\big(j^ks\big)^*\alpha=\int_M\big(j^{k+1}s\big)^*h (\alpha).
 \]
\end{Theorem}
\begin{proof} It follows from the formula
 \[
 T_xj^ks=T_xj^ks\circ T_{j^{k+1}_xs}\pi_{k+1}\circ T_xj^{k+1}s,
 \]
 that holds for every $x\in M$and $s\in\Gamma\pi$.
\end{proof}

It is immediate to prove the desired equivalence.

\begin{Corollary} The classical variational problem specified by the Lagrangian density $\cL_{\rm EH}$ \linebreak on~$J^2{\tau}_\Sigma$ and the variational problem $\big(J^1{\tau}_\Sigma,\cL_{\rm EH}^{(1)},\mathcal{I}^\Sigma_{{\rm con}}\big)$ have the same set of extremals.
\end{Corollary}
\begin{proof}From Lemma \ref{lem:DensityAndLagFormEH} and using Theorem \ref{thm:IntegralsAndHorizontalization}, we see that $g\colon M\rightarrow\Sigma$ is an extremal for the action integral
 \[
 g\mapsto\int_M\big(j^2g\big)^*\cL_{\rm EH}
 \]
 if and only if it is an extremal for the action integral
 \[
 g\mapsto\int_M\big(j^1g\big)^*\cL_{\rm EH}^{(1)},
 \]
 as required.
\end{proof}

As usual \cite{GotayCartan}, the equations of motion of this variational problem can be lifted to a space of forms on $J^1\tau_\Sigma$. Let us define the affine subbundle
\[
 W_{\rm EH}:=\cL_{\rm EH}^{(1)}+{I}^\Sigma_{{\rm con},2}\subset\wedge^m\big(J^1\tau_\Sigma\big).
\]
Here, for every $j_x^1s\in J^1\tau_\Sigma$,
\begin{gather*}
 I^\Sigma_{{\rm con},2}\big|_{j_x^1s} =\mathcal{L}\big\{{\alpha}_{s(x)}\circ\big(T_{j_x^1s} (\tau_\Sigma )_{10}-T_xs\circ T_{j_x^1s} (\tau_\Sigma )_1\big)\wedge\beta\colon\\
 \hphantom{I^\Sigma_{{\rm con},2}\big|_{j_x^1s} =\mathcal{L}\big\{}{} {\alpha}_{s(x)}\in T_{s(x)}^*\Sigma,\beta\in\big(\Lambda^{m-1}_1\big(J^1\tau_\Sigma\big)\big)_{j_x^1s}\big\}
\end{gather*}
is the corresponding fiber for the contact subbundle on $J^1\tau_\Sigma$. The canonical map will be denoted by
\[
 \tau_{\rm EH}\colon \ W_{\rm EH}\to J^1\tau_\Sigma.
\]

We will indicate with $\lambda_{\rm EH}$ the pullback of the canonical $m$-form on $\wedge^m\big(J^1\tau_\Sigma\big)$ to $W_{\rm EH}$. Then we have a result analogous to Proposition~\ref{prop:FieldTheoryEqsWL} in the context of (first order) Einstein--Hilbert formulation.
\begin{Proposition}\label{prop:FieldTheoryEqsEinsteinHilbertCase}
A section $s\colon U\subset M\rightarrow J^1\tau_\Sigma$ is a critical holonomic section for the variational problem $\big(J^1{\tau}_\Sigma,\cL_{\rm EH}^{(1)},\mathcal{I}^\Sigma_{{\rm con}}\big)$ if and only if there exists a section $\Gamma\colon U\subset M\rightarrow W_{\rm EH}$ such that
\begin{enumerate}\itemsep=0pt
\item[$1)$] $\Gamma$ covers $s$, i.e., $\tau_{\rm EH}\circ\Gamma=s$, and
\item[$2)$] $\Gamma^* (X\lrcorner {\rm d}\lambda_{\rm EH} )=0$, for all $X\in\mathfrak{X}^{V ( (\tau_\Sigma\ )_1\circ\tau_{\rm EH} )}(W_{\rm EH})$.
\end{enumerate}
\end{Proposition}

\begin{Remark} This proposition provides us with a unified formalism for Einstein--Hilbert gravity, based on the first order formulation. For the corresponding formalism associated to the second order formulation, see~\cite{doi:10.1063/1.4998526}.
\end{Remark}

\section{Contact bundle decomposition for Palatini gravity}\label{sec:cont-bundle-decomp}

We return here to the discussion initiated in Section~\ref{sec:reduc-bundle-palat}, regarding the splitting of the quotient~$W_{\rm PG}/K$; it will be shown that the connection $\omega_K$ is useful also for the identification of elements in~$W_{\rm PG}/K$ that can be seen as elements of an space of forms. Intuitively, it means that the associated degrees of freedom can be interpreted as multimomenta.

In order to perform this identification, we will recall some general facts regarding the decomposition induced for the connection $\omega_K$ \cite{Capriotti2019} on the bundle of forms $W_{\rm PG}$ defined in equation~\eqref{eq:BundleOfFormsOnT0}. The contact structure on $J^1\tau$ gives rise to the contact subbundle on $\cT_0$ given by
\begin{gather}
 I^m_{{\rm con},2}\big|_{j_x^1s} =\mathcal{L}\big\{{\alpha}_{s(x)}\circ(T_{j_x^1s}\tau_{10}'-T_xs\circ T_{j_x^1s}\tau_1')\wedge\beta\colon \nonumber\\
 \hphantom{I^m_{{\rm con},2}\big|_{j_x^1s} =\mathcal{L}\big\{}{} {\alpha}_{s(x)}\in T_{s(x)}^*(LM),\beta\in\big(\Lambda^{m-1}_1 (\cT_0 )\big)_{j_x^1s}\big\},\label{eq:ContactFields}
\end{gather}
where $\mathcal{L}$ indicates linear closure. There is an splitting of $I^m_{{\rm con},2}$ induced by the choice of a~connection on the principal bundle $p^{LM}_K\colon LM\to \Sigma$. Its construction proceeds as follows. We denote by $\omega_K\in\Omega^1(LM,\kf)$ the chosen connection and consider the following splitting of the cotangent bundle:
\[
T^*(LM)=\big(p^{LM}_K\big)^* (T^*\Sigma ) \oplus (LM\times \kf^*).
\]
The identification is obtained as follows:
\begin{align*}
\big(p^{LM}_K\big)^* (T^*\Sigma ) \oplus (LM\times \kf^*)&\to T^*(LM),\\
(e,\widehat{\alpha}_{[u]},\sigma)&\mapsto \alpha_u=\widehat{\alpha}_{[u]}\circ T_up_K^{LM}+\langle\sigma,\omega_K(\cdot)\rangle.
\end{align*}
Accordingly, we have an splitting of contact bundle~\eqref{eq:ContactFields}
\begin{gather*}%\label{eq:splittingcontact_LFT}
I^m_{{\rm con},2}=\widetilde{I^m_{{\rm con},2}}\oplus I^m_{\kf^*,2},
\end{gather*}
with
\begin{gather*}
 \widetilde{I^m_{{\rm con},2}}\big|_{j_x^1s} =\mathcal{L}\big\{\widehat{\alpha}_{[s(x)]}\circ T_{s(x)}p_K^{LM}\circ\big(T_{j_x^1s}\tau_{10}'-T_xs\circ T_{j_x^1s}\tau_1'\big)\wedge\beta\colon\\
 \hphantom{\widetilde{I^m_{{\rm con},2}}\big|_{j_x^1s} =\mathcal{L}\big\{}{}
 \widehat{\alpha}_{[s(x)]}\in T_{[s(x)]}^*\Sigma,\beta\in\big(\Lambda^{m-1}_1\cT_0\big)_{j_x^1s}\big\},\\
I_{\kf^*,2}^m\big|_{j_x^1s} =\big\{\big\langle\sigma\stackrel{\wedge}{,}\omega_K\circ\big(T_{j_x^1s}\tau_{10}'-T_xs\circ T_{j_x^1s}\tau_1'\big)\big\rangle\colon \sigma\in\big(\Lambda^{m-1}_1\cT_0\otimes\kf^*\big)_{j_x^1s}\big\}.
\end{gather*}
The symbol $\langle\cdot\stackrel{\wedge}{,} \cdot\rangle$ denotes the natural contraction, defined as follows: For elements of the form $\alpha_1\otimes\nu$, $\alpha_2\otimes \eta$ with $\nu,\eta\in\kf$ and $\alpha_1,\alpha_2$ forms, we have $\langle\alpha_1\otimes\nu \stackrel{\wedge}{,} \alpha_2\otimes\eta \rangle=\langle\nu,\eta\rangle \alpha_1\wedge\alpha_2$. For a~general element in the linear closure, the definition extends linearly.

We can split our metricity subbundle $I^m_{\rm PG}$ using the inclusion
\[
 I^m_{\rm PG}\subset I^m_{{\rm con},2},
\]
namely
\[
 I^m_{\rm PG}=\big(I^m_{\rm PG}\cap\widetilde{I^m_{{\rm con},2}}\big)\oplus\big(I^m_{\rm PG}\cap I^m_{\kf^*,2}\big).
\]
But we have the following fact.
\begin{Lemma}\label{lem:cont-bundle-decomp}
 For every $j_x^1s\in\cT_0$
 \[
 I^m_{\rm PG}\subset\widetilde{I^m_{{\rm con},2}}.
 \]
\end{Lemma}
\begin{proof}Let us work in the coordinates considered above; therefore, we have equation~\eqref{eq:E_Projection} for the projector $Tp_K^{LM}$ and also
 \[
 T_{j_x^1s}\tau_{10}'-T_xs\circ T_{j_x^1s}\tau_1'=\frac{\partial}{\partial e^\mu_k}\otimes\big({\rm d}e^\mu_k-e^\mu_{k\alpha}{\rm d}x^\alpha\big).
 \]
 Then
 \begin{gather*}
 T_{s(x)}p_K^{LM}\circ \big(T_{j_x^1s}\tau_{10}'-T_xs\circ T_{j_x^1s}\tau_1'\big) =T_{s(x)}p_K^{LM}\left(\frac{\partial}{\partial e^\mu_k}\right)\otimes\big({\rm d}e^\mu_k-e^\mu_{k\alpha}{\rm d}x^\alpha\big)\\
 \qquad{} =\frac{\partial}{\partial g^{\rho\sigma}}\otimes\big[\eta^{kq}\big(e_q^\rho {\rm d}e_k^\sigma+e_q^\sigma {\rm d}e_k^\rho\big)-\eta^{kq}\big(e_q^\rho e_{k\alpha}^\sigma+e_q^\sigma e_{k\alpha}^\rho\big){\rm d}x^\alpha\big]\\
 \qquad{} =\frac{\partial}{\partial g^{\rho\sigma}}\otimes\big[{\rm d}g^{\rho\sigma}-\big(g^{\rho\beta}e_\beta^ke^\sigma_{k\alpha} +g^{\sigma\beta}e_\beta^ke^\rho_{k\alpha}\big){\rm d}x^\alpha\big]\\
 \qquad{} =\frac{\partial}{\partial g^{\rho\sigma}}\otimes\big[{\rm d}g^{\rho\sigma}+\big(g^{\rho\beta}\Gamma^\sigma_{\beta\alpha} +g^{\sigma\beta}\Gamma^\rho_{\beta\alpha}\big){\rm d}x^\alpha\big]
 \end{gather*}
 that will define to a set of generators of the bundle $I^m_{\rm PG}$ (see equation~\eqref{eq:GammaInLocalTerms}).
\end{proof}

This result is compatible with the fact that the whole subbundle $W_{\rm PG}$ is in the zero level set for the momentum map. We will return to that below.

\section{First order Einstein--Hilbert Lagrangian as Routhian}\label{sec:routhian}

As we mentioned in the introductory sections, a crucial role in Routh reduction is played by the \emph{Routhian}, which replaces the Lagrangian in determining the dynamics of the mechanical system. This replacement is unavoidable, because the dynamics for the Routh reduced problem should happen on a level set for the momentum map of the theory
\[
 J\colon \ TQ\rightarrow\mathfrak{g}^*,
\]
and the constraints imposed by this requirement must be included in the Lagrangian.

To fix ideas, let us briefly describe how this construction proceeds in the case of a classical variational problem
\[
 \big({\rm pr}_1\colon \mathbb{R}\times TQ\to\mathbb{R},L{\rm d}t,\langle {\rm d}q-\dot{q}{\rm d}t\rangle \big)
\]
for a mechanical system with configuration space~$Q$ and a symmetry described by a Lie group $G$ acting freely on $Q$. The idea is to incorporate the constraints imposed by the momentum map into the Lagrangian through a family of Lagrange multipliers; this prescription tells us that the Routhian becomes
\[
 R_\mu:=L-\langle \mu,\omega_Q\rangle
\]
where $\mu\in\mathfrak{g}^*$ is the level chosen for the momentum map and~$\omega_Q$ is a connection form for the principal bundle $p_G^Q\colon Q\to Q/G$. A similar construction can be done when working for classical variational problems describing a first order field theory~\cite{Capriotti2019}.

Now, let us try to particularize these considerations for the case of Palatini gravity. The Routhian form is expected to coincide with the Lagrangian $\cL_{\rm PG}$ because~$J\equiv0$. Nevertheless, as we have stressed above, this result would be valid if we were working within the range of~\cite{Capriotti2019}; on the contrary, the variational problem describing Palatini gravity is not reached by these results, and so this should be properly verified in this particular case.

First, we write $\overline{p}\colon {\rm Lin}{\big(\tau_\Sigma^*TM,\widetilde{\kf}\big)}\to\Sigma$ for the obvious projection. In principle, the bundle ${\rm Lin}{\big(\tau_\Sigma^*TM,\widetilde{\kf}\big)}$ would be the field bundle for the reduced system; nevertheless, we will show in Lemma~\ref{lem:routhian_pr1_horizontal} that the Routhian, namely, the Lagrangian form for this reduced system, will be horizontal for the projection onto the jet space of the base bundle~$\Sigma$.

In particular, one can consider the map:
\begin{align*}
 q\colon \ J^1\big(\tau_\Sigma\circ\overline{p}\big)&\longto J^1\tau_\Sigma\times {\rm Lin}{\big(\tau_\Sigma^*TM,\widetilde{\kf}\big)},\\
 j_x^1\sigma&\longmapsto\big(j_x^1\big(\overline{p}\circ\sigma\big),\sigma(x)\big)
\end{align*}
projecting onto the quotient bundle for Palatini gravity. So, we can formulate the reduced system as a first order field theory by taking the bundle $\tau_\Sigma\circ\overline{p}\colon {\rm Lin}{\big(\tau_\Sigma^*TM,\widetilde{\kf}\big)}\to M$ as the basic field bundle. Nevertheless, there are some identifications that will permit us to simplify this basic bundle further.

In order to proceed, let use the connection $\omega_K$ to define the maps fitting in the following diagram:
\begin{equation*}
 \begin{tikzcd}[ampersand replacement=\&, column sep=1.2cm, row
 sep=.9cm]
 {\cT_0} \arrow[swap]{d}{p_{K}^{J^1\tau}\big|_{\cT_0}} \arrow{r}{f_\omega} \&
 {{\rm Lin}{\big(\tau_\Sigma^*TM,\widetilde{\kf}\big)}} \&
 {J^1\big(\tau_\Sigma\circ\overline{p}\big)}
 \arrow[swap]{l}{\big(\tau_\Sigma\circ\overline{p}\big)_{10}}
 \arrow{d}{q}
 \\
 {\cT_0/K}
 \arrow[swap]{rr}{g_\omega:= \overline{\Upsilon}_\omega |_{\cT_0/K}}
 \&
 \&
 {J^1\tau_\Sigma\times{\rm Lin}{\big(\tau_\Sigma^*TM,\widetilde{\kf}\big)}}
 \end{tikzcd}
\end{equation*}
The definitions are as follows:
\begin{align*}
f_\omega\colon \ \cT_0& \longto {\rm Lin}{\big(\tau_\Sigma^*TM,\widetilde{\kf}\big)},\\
j_x^1s&\longmapsto [s(x),\omega_K\circ T_xs ]_K,\\
 g_\omega\colon \ \cT_0/{K}& \longto J^1\tau_\Sigma\times{\rm Lin}{\big(\tau_\Sigma^*TM,\widetilde{\kf}\big)},\\
\big[j_x^1s\big]_{K} &\longmapsto\big(j_x^1\big(p_K^{LM}\circ s\big), [s(x),\omega_K\circ T_xs ]_K\big).
\end{align*}

The map $g_\omega$ is the identification from Corollary~\ref{cor:identification}. Since the Lagrangian form $\cL_{\rm PG}$ is basic for the projection $p_{\cT_0}^K\colon \cT_0\to\cT_0/K$, it defines a reduced form on $\cT_0/K$ which can be seen as a~form on $J^1\tau_\Sigma\times\mathop{\rm Lin}{\big(\tau_\Sigma^*TM,\widetilde{\kf}\big)}$. We will denote it by $\overline{\cL}_{\rm PG}$:
\[
 \big(g_\omega\circ p_{K}^{J^1\tau}\big)^*\overline{\cL}_{\rm PG}=\cL_{\rm PG},\qquad \overline{\cL}_{\rm PG} \in \Omega_2^m\big(J^1\tau_\Sigma\times\mathop{\rm Lin}{\big(\tau_\Sigma^*TM,\widetilde{\kf}\big)}\big).
\]
\begin{Definition}
 The $m$-form $\overline{\cL}_{\rm PG}\in\Omega^m\big(J^1\tau_\Sigma\times\mathop{\rm Lin}{\big(\tau_\Sigma^*TM,\widetilde{\kf}\big)}\big)$ is the \emph{Routhian} for the variational problem $\big(\cT_0,\cL_{\rm PG},\mathcal{I}_{\rm PG}^m\big)$.
\end{Definition}

Then, we are ready to prove a characteristic property for the Routhian associated to the reduction of Palatini gravity.
\begin{Lemma}\label{lem:routhian_pr1_horizontal}
 The Routhian $\overline{\cL}_{\rm PG}$ is ${\rm pr}_1$-horizontal, where
 \[
 {\rm pr}_1\colon \ J^1\tau_\Sigma\times\mathop{\rm Lin}{\big(\tau_\Sigma^*TM,\widetilde{\kf}\big)}\longto J^1\tau_\Sigma
 \]
 is the projection onto the first factor of the fibred product.
\end{Lemma}
\begin{proof} It follows from equations~\eqref{eq:GammaInTermsQuotient} and~\eqref{eq:GammaInTermsQuotientJetVariables} that
 \begin{equation}\label{eq:EinsteinLagVsPalatiniLag}
 {\rm pr}_1^*\cL_{\rm EH}^{(1)}=\overline{\cL}_{\rm PG},
 \end{equation}
 as required.
\end{proof}

In short, Routhian $\overline{\cL}_{\rm PG}$ does not depend on the fiber coordinates $A^\sigma_{\mu\rho}$ of the bundle \[
\overline{p}\colon \ {\rm Lin}{\big(\tau_\Sigma^*TM,\widetilde{\kf}\big)}\to\Sigma;
\] it is just the pullback along ${\rm pr}_1$ of the first order Lagrangian for Einstein--Hilbert gravity.

In the usual Routh reduction, the reduced Routhian is a $m$-form on $J^1\big(\tau_\Sigma\circ\overline{p}\big)$; in this case, Lemma~\ref{lem:routhian_pr1_horizontal} allows us to consider the form $\cL_{\rm EH}^{(1)}$ on $J^1\tau_\Sigma$ as the Routhian. Therefore, we can forget about the degrees of freedom associated to the factor $\mathop{{\rm Lin}}{\big(\tau_\Sigma^*TM,\widetilde{\kf}\big)}$ in the quotient bundle, and take as the quotient bundle for Palatini gravity the jet bundle~$J^1\tau_\Sigma$; this is the way in which we will proceed from this point.

\section[Einstein--Hilbert gravity as Routh reduction of Palatini gravity]{Einstein--Hilbert gravity as Routh reduction\\ of Palatini gravity}\label{sec:einst-hilb-grav}

We will devote the present section to establish the two main results of the article, namely, Theo\-rem~\ref{thm:routh-reduct-palat} regarding reduction of Palatini gravity and Theorem~\ref{thm:conn-induc-pullb} dealing with reconstruction of metrics verifying Einstein equations of gravity. The strategy, as we mention in the introduction, is to compare the equations of motion (lifted to the corresponding spaces of forms~$W_{\rm EH}$ and~$W_{\rm PG}$) in a bundle containing every relevant degree of freedom; this role is played below by a pullback bundle of the bundle~$W_{\rm EH}$ along a suitable map. So, let us define
\[
 F_\omega:={\rm pr}_1\circ g_\omega\circ p_K^{J^1\tau}\colon \ \cT_0\longto J^1\tau_\Sigma,
\]
namely
\[
 F_\omega\left(j_x^1s\right)=j_x^1\left(p_K^{LM}\circ s\right)
\]
for every $j_x^1s\in\cT_0$. Then we have the diagram
\begin{equation}\label{eq:DiagramEHVsPG}
 \begin{tikzcd}[row sep=1.3cm,column sep=1cm]
 W_{\rm PG}
 \arrow[swap]{dr}{\pi_{\rm PG}}
 \arrow[hook]{r}{}
 &
 \wedge_2^m\left(T^*\cT_0\right)
 \arrow{d}{\overline{\tau}^m_{\cT_0}}
 &
 F_\omega^*\left(W_{\rm EH}\right)
 \arrow{dl}{{\rm pr}_1^\omega}
 \arrow[swap]{l}{\widetilde{F_\omega}}
 \arrow{r}{{\rm pr}_2^\omega}
 &
 W_{\rm EH}
 \arrow{d}{\pi_{\rm EH}}
 \\
 &
 \cT_0
 \arrow[swap]{rr}{F_\omega}
 &
 &
 J^1\tau_\Sigma
 \end{tikzcd}
\end{equation}
where
\[
 \widetilde{F_\omega}\colon \ F_\omega^* (W_{\rm EH} )\longto\wedge_2^m (T^*\cT_0 )\colon \ \big(j_x^1s,\rho\big)\mapsto\rho\circ T_{j_x^1s}F_\omega
\]
and
\begin{gather*}
 {\rm pr}_1^\omega\colon \ F_\omega^* (W_{\rm EH} )\longto \cT_0,\qquad{\rm pr}_2^\omega\colon \ F_\omega^* (W_{\rm EH} )\longto W_{\rm EH}
\end{gather*}
are the canonical projections of the pullback bundle.

\begin{Remark}\label{rem:LocalFomega} Let us give a local version of the map $F_\omega$; recall that locally $\cT_0$ is described by the set of equations
 \[
 e^k_\sigma e^\nu_{k\rho}=e^k_\rho e^\nu_{k\sigma},
 \]
 where $\big(x^\mu,e_k^\sigma,e^\rho_{k\mu}\big)$ is a set of adapted coordinates induced by a coordinate chart $ (\phi= (x^\mu ),U )$ on $M$. In these coordinates, the canonical quotient map $p_{K}^{LM}\colon LM\to\Sigma$ is given by
 \[
 p_{K}^{LM}\big(x^\mu,e_k^\sigma\big)=\big(x^\mu,\eta^{kl}e_k^\rho e_l^\sigma\big);
 \]
 accordingly, map $F_\omega$ will become
 \[
 F_\omega\big(x^\mu,e_k^\sigma,e^\rho_{k\mu}\big)=\big(x^\mu,\eta^{kl}e_k^\sigma e_l^\rho,\eta^{kl}\big(e^\sigma_{k\mu}e_l^\rho+e^\sigma_{k}e_{l\mu}^\rho\big)\big).
 \]
\end{Remark}

\begin{Lemma}\label{lem:FwVsWPG} The bundle map $\widetilde{F_\omega}$ is an affine bundle isomorphism on $\cT_0$ between $W_{\rm PG}$ and $F_\omega^* (W_{\rm EH})$.
\end{Lemma}
\begin{proof}
 It is consequence of equation~\eqref{eq:EinsteinLagVsPalatiniLag} and Proposition~\ref{prop:Metricity-Horizontal}.
\end{proof}

We will use diagram~\eqref{eq:DiagramEHVsPG} as a means to compare the equations of motion of Palatini gravity and Einstein--Hilbert gravity; the idea is to use Propositions~\ref{prop:FieldTheoryEqsWL} and~\ref{prop:FieldTheoryEqsEinsteinHilbertCase} in order to represent these equations in terms of the spaces of forms $W_{\rm PG}$ and $W_{\rm EH}$ respectively, and to pull them back to the common space $F_\omega^* (W_{\rm EH})$. Crucial to this strategy is the following result.

\begin{Proposition}\label{prop:RelationBetweenCanonicalForms}
 The following relation holds
 \[
 \widetilde{F_\omega}^*\lambda_{\rm PG}=\big({\rm pr}_2^\omega\big)^*\lambda_{\rm EH}.
 \]
\end{Proposition}
\begin{proof} Let $\big(j_x^1s,\rho\big)\in F^*_\omega (W_{\rm EH})$ be an arbitrary element in this pullback bundle; then using diagram~\eqref{eq:DiagramEHVsPG} we will have that
 \begin{align*}
 \lambda_{\rm PG}\big|_{\rho\circ T_{j_x^1s}F_\omega}\circ T_{(j_x^1s,\rho)}\widetilde{F_\omega}&=\big(\rho\circ T_{j_x^1s}F_\omega\big)\circ T_{\rho\circ T_{j_x^1s}F_\omega}\overline{\tau}^m_{\cT_0}\circ T_{(j_x^1s,\rho)}\widetilde{F_\omega}\\
 &=\big(\rho\circ T_{j_x^1s}F_\omega\big)\circ T_{(j_x^1s,\rho)}{\rm pr}_1^\omega
 =\rho\circ T_\rho\pi_{\rm EH}\circ T_{(j_x^1s,\rho)}{\rm pr}_2^\omega\\
 & =\lambda_{\rm EH}|_\rho\circ T_{(j_x^1s,\rho)}{\rm pr}_2^\omega,
 \end{align*}
 using the fact that $\pi_{\rm EH}\colon W_{\rm EH}\to J^1\tau_\Sigma$ is the restriction of the canonical projection
 \[
 \overline{\tau}^m_{J^1\tau_\Sigma}\colon \ \wedge^m_2\big(T^*J^1\tau_\Sigma\big)\longto J^1\tau_\Sigma
 \]
 to $W_{\rm EH}$. This identity proves the proposition.
\end{proof}

\subsection{Routh reduction of Palatini gravity}\label{sec:routh-reduct-palat}

We are now ready to prove the first result on Routh reduction of Palatini gravity; essentially, we will prove that any section obeying the equations of motion for Palatini gravity projects along the map $F_\omega$ to a holonomic section obeying Einstein--Hilbert equations of motion.

\begin{Theorem}\label{thm:routh-reduct-palat}
 Let $\widehat{Z}\colon U\subset M\to \cT_0$ be a section that obeys the Palatini gravity equations of motion. Then the section
 \[
 F_\omega\circ\widehat{Z}\colon \ U\to J^1\tau_\Sigma
 \]
 is holonomic and obeys the Einstein--Hilbert gravity equations of motion.
\end{Theorem}
\begin{proof}The idea of the proof is encoded in the following diagram
 \begin{equation}\label{eq:EqsMotionEHPG}
 \begin{tikzcd}[row sep=1.7cm,column sep=1.1cm]
 W_{\rm PG}
 \arrow{d}{\pi_{\rm PG}}
 &
 F_\omega^*(W_{\rm EH})
 \arrow[swap]{l}{\widetilde{F_\omega}}
 \arrow{dl}{{\rm pr}_1^\omega}
 \arrow{r}{{\rm pr}_2^\omega}
 &
 W_{\rm EH}
 \arrow[swap]{d}{\pi_{\rm EH}}
 \\
 \cT_0
 \arrow[near end]{rr}{F_\omega}
 \arrow{dr}{\tau_1}
 &
 &
 J^1\tau_\Sigma
 \arrow{dl}{(\tau_\Sigma)_1}
 \\
 &
 M
 \arrow[dashed,bend left=25]{ul}{\widehat{Z}}
 \arrow[dashed,bend left=75]{uul}{\Gamma}
 \arrow[dashed,bend right=10,swap,near end,crossing over]{uu}{\Gamma'}
 \arrow[dashed,bend right=70,swap]{uur}{\widetilde{\Gamma}}
 &
 \end{tikzcd}
 \end{equation}
 Using Proposition~\ref{prop:FieldTheoryEqsWL}, we construct $\Gamma\colon U\to W_{\rm PG}$ out of $\widehat{Z}$; the Palatini gravity equations of motion will become
 \[
 \Gamma^* (Z\lrcorner {\rm d}\lambda_{\rm PG} )=0
 \]
 for any $Z\in\mathfrak{X}^{V (\tau_1\circ\pi_{\rm PG} )} (W_{\rm PG} )$. Using Lemma~\ref{lem:FwVsWPG} we can define
 \[
 \Gamma':=\big(\widetilde{F}_\omega\big)^{-1}\circ\Gamma\colon \ U\longto F_\omega^* (W_{\rm EH} );
 \]
 then the Palatini equations of motion translate into
 \[
 ({\Gamma}' )^* \big({Z}'\lrcorner {\rm d}\widetilde{F_\omega}^*\lambda_{\rm PG}\big)=0
 \]
 for any ${Z}'\in\mathfrak{X}^{V(\tau_1\circ{\rm pr}_1^{\omega})}(F^*(W_{\rm EH}))$. Then using Proposition~\ref{prop:RelationBetweenCanonicalForms} and the fact that ${\rm pr}_2^\omega\colon F^*_\omega (W_{\rm EH})\allowbreak \to W_{\rm EH}$ is a submersion, we can conclude that the section
 \[
 \widetilde{\Gamma}:={\rm pr}_2^\omega\circ\Gamma'\colon \ U\to W_{\rm EH}
 \]
 obeys the equations of motion
 \[
 \widetilde{\Gamma}^*\big(\widetilde{Z}\lrcorner {\rm d}\lambda_{\rm EH}\big)=0,
 \]
 where $\widetilde{Z}\in\mathfrak{X}^{V((\tau_\Sigma)_1\circ\pi_{\rm EH})}(W_{\rm EH})$ is an arbitrary vertical vector field on $W_{\rm EH}$. Also, using diagram~\eqref{eq:EqsMotionEHPG}, we have that
 \[
 \pi_{\rm EH}\circ\widetilde{\Gamma}=\pi_{\rm EH}\circ{\rm pr}_2^\omega\circ\Gamma'=F_\omega\circ{\rm pr}_1^\omega\circ\Gamma'=F_\omega\circ\widehat{Z}.
 \]
 The theorem then follows from Proposition~\ref{prop:FieldTheoryEqsEinsteinHilbertCase}.
\end{proof}

\begin{Remark}[reduction theorem in local coordinates] Let us look at the reduction theorem in local terms. First, we have that the equations of motion on $J^1\tau$ can be written as
 \begin{gather}
 \frac{\partial\Gamma^\beta_{\mu\nu}}{\partial x^\beta}-\frac{\partial\Gamma^\beta_{\mu\beta}}{\partial x^\nu}+\Gamma^\beta_{\sigma\beta}\Gamma^\sigma_{\mu\nu}-\Gamma^\beta_{\sigma\nu}\Gamma^\sigma_{\mu\beta}=0,\nonumber\\\label{eq:EinsteinPalatiniEqs}
 \Gamma^\sigma_{\mu\nu}=-e_\mu^k e^\sigma_{k\nu},\nonumber\\
 e_\mu^k e_{k\nu}^\sigma -e_\nu^k e_{k\mu}=0.\label{eq11-5}
 \end{gather}
 On the other hand, from Remark~\ref{rem:LocalFomega} we have that
 \[
 F_\omega\big(x^\mu,e_k^\sigma,e^\rho_{k\mu}\big)=\big(x^\mu,\eta^{kl}e_k^\sigma e_l^\rho,\eta^{kl}\big(e^\sigma_{k\mu}e_l^\rho+e^\sigma_{k}e_{l\mu}^\rho\big)\big).
 \]
 Thus, Theorem \ref{thm:routh-reduct-palat} tells us that any local section
 \[
 \widehat{Z}\big(x^\mu\big)=\big(x^\mu,e^\mu_k(x),e^\mu_{k\nu}(x)\big)
 \]
 obeying equations~\eqref{eq11-5} will gives rise to a solution of Einstein equations of motion when composed with~$F_\omega$.
\end{Remark}

\subsection{\dots\ and reconstruction}\label{sec:...and-reconstr}

We will give now a (somewhat partial) converse to Theorem \ref{thm:routh-reduct-palat}. That is, given a section $\zeta\colon U\subset M\to\Sigma$ such that $j^1\zeta\colon U\to J^1\tau_\Sigma$ is extremal for the Einstein--Hilbert variational problem, find a section
\[
 \widehat{Z}\colon \ U\to\cT_0
\]
such that $F_\omega\circ\widehat{Z}=j^1\zeta$ and $\widehat{Z}$ is an extremal for the Palatini variational problem. From Fig.~\ref{fig:MapsInvolvedRouth} it is clear that we need to lift the section~$j^1\zeta$ through the quotient map $ p_K^{J^1\tau}\big|_{\cT_0}\colon \cT_0\to\cT_0/K$, which has the structure of a principal bundle on~$\cT_0/K$. It is clear that any principal bundle can be trivialized by a convenient restriction of the base space. As discussed in~\cite{Capriotti2019}, it is not the way in which this kind of reconstruction problems are solved. Rather, the problem of lifting sections along the projection in a principal bundle is reduced to the problem of deciding if certain connection is flat; moreover, it is expected that this connection is related to the connection used to define the Routhian. We will present in this section a theory of reconstruction along these lines. With this goal in mind, we will recall here some of the details developed in~\cite{Capriotti2019}; for proofs we refer to the original article. We begin with a pair of diagrams~\eqref{dia:covering}:
\begin{equation}\label{dia:covering}
 \begin{tikzcd}[column sep=1cm, row
 sep=1.4cm]
 P \arrow[rr,"p^P_G"]\arrow[dr,swap,"\pi"] & & [1ex]P/G \arrow[dl,"\overline{\pi}"]\\
 & M \arrow[ur,dashed,bend right=45,swap,"\zeta"]\arrow[ul,dashed,bend left=45,"s"]&
 \end{tikzcd}\hspace{3em}
 \begin{tikzcd}[column sep=1.8cm, row
 sep=1.4cm]
 \zeta^*P \arrow[r,"{\rm pr}_2"]\arrow[d,swap,"{\rm pr}_1"] &[1ex] P\arrow[d,"p^P_G"]\\
 M\arrow[r,swap,"\zeta"] & P/G
 \end{tikzcd}
\end{equation}

Then we have the following result.

\begin{Lemma}\label{lem:TrivialToSections} There exists a section $s\colon M\rightarrow P$ covering the section $\zeta\colon M\to P/G$ if and only if $\zeta^*P$ is a trivial bundle.
\end{Lemma}

Using that $\zeta^*P$ is a principal bundle, being trivial can be characterized in terms of a flat connection~\cite{KN1}:

\begin{Theorem}\label{thm:conn-induc-pullb} Let $\pi\colon P\rightarrow M$ be a $G$-principal bundle with $M$ simply connected. Then~$P$ is trivial if and only if there exists a flat connection on~$P$.
\end{Theorem}

If $M$ is not simply connected, then one can ask for a flat connection with trivial holonomy and obtain a similar result. For the sake of simplicity, we will assume that~$M$ is simply connected to apply Theorem~\ref{thm:conn-induc-pullb} when needed. For later use, we also observe that the section constructed in the proof of Theorem~\ref{thm:conn-induc-pullb} has horizontal image with respect to the given connection.

We wish now to apply the previous discussion to the case of the bundle $p_K^{\cT_0}:= p_K^{J^1\tau}\big|_{\cT_0}\colon \cT_0\to\cT_0/K$. We have the situation depicted in diagram~\eqref{dia:covering2} (left): $Z\colon M\to\cT_0/K$ is a given section and $\zeta\colon M\to \Sigma$ is the induced section. The basic question we want to address is whether there exists a section $\widehat{Z}\colon M\rightarrow \cT_0$ {such that} $p_K^{\cT_0}\circ\widehat{Z}=Z$:
\begin{equation}\label{dia:covering2}
 \begin{tikzcd}[column sep=1cm, row
 sep=1cm]
 \cT_0 \arrow[rr,"p^{\cT_0}_K"]\arrow[d,"\tau_{10}'"] &&\cT_0/K\arrow[d,swap,"\overline{\tau}'_{10}"]\\
 LM\arrow[rr,"p^{LM}_K"]\arrow[dr,"\tau'"] & & \Sigma \arrow[dl,swap,"\tau_\Sigma"] \\
 & M\arrow[ur,swap,dashed,"\zeta",bend right=25]\arrow[uul,dashed,"\widehat{Z}",bend left=85]\arrow[uur,swap,dashed,"Z",bend right=85]&
 \end{tikzcd}
 \hspace{1em}
 \begin{tikzcd}[column sep=2cm, row
 sep=1.8cm]
 \zeta^*(LM) \arrow[r,"{\rm pr}_2"]\arrow[d,swap,"{\rm pr}_1"] &LM\arrow[d,"p_K^{LM}"]\\
 M\arrow[r,swap,"\zeta"] & \Sigma
 \end{tikzcd}
\end{equation}
Now, using the fact that
\[
 \cT_0/K\simeq\Sigma\times C_0(LM),\qquad\cT_0\simeq LM\times C_0(LM)
\]
we have that $Z= (\zeta,\Gamma)$ is composed by the metric plus the Levi-Civita connection $\Gamma$; therefore, we will have that
\[
 \widehat{Z}=\big(\widehat{\zeta},\Gamma\big)
\]
where $\widehat{\zeta}\colon M\to LM$ is some lift of the section $\zeta\colon M\to\Sigma$. Then, we can then construct the pullback bundle $\zeta^*(LM)$ (diagram~\eqref{dia:covering2}, right) and particularize Lemma~\ref{lem:TrivialToSections} to conclude the following:
\begin{Lemma}\label{lem:FlatConnectionIffSectionJ1pi} Assume that $M$ is simply connected. If $\zeta^*(LM)$ admits a flat connection then there exists a section $\widehat{Z}\colon M\rightarrow \cT_0$ such that
 \[
\big(p_K^{LM}\circ\tau_{10}'\big)\circ\widehat{Z}=\zeta
 \]
 and $\widehat{Z}^*\omega_\pf=0$. Conversely, every such section gives rise to a flat connection on $\zeta^* (TM)$.
\end{Lemma}
\begin{proof} Because $\zeta^*(LM)$ is a $K$-principal bundle, Theorem \ref{thm:conn-induc-pullb} and Lemma \ref{lem:TrivialToSections} allow us to find a section $\widehat{\zeta}\colon M\to LM$ iff there exists a flat connection on it. Thus if $\omega^\zeta$ is flat, we can construct a lift $\widehat{\zeta}\colon M\to LM$ for $\zeta$ and so
 \[
 \widehat{Z}=\big(\widehat{\zeta},\Gamma\big)
 \]
 is the desired lift to $\cT_0$, where $\Gamma\colon M\to C(LM)$ is the Levi-Civita connection for $\zeta$.

 Conversely, let us suppose that we have a lift
 \[
 \widehat{\zeta}:=\tau_{10}'\circ\widehat{Z}\colon \ M\to LM
 \]
 for the metric $\zeta\colon M\to\Sigma$. Recall that, for every $(x,u)\in\zeta^*(LM)$,
 \[
 T_{(x,u)}\zeta^*(LM)=\big\{(v_x,V_u)\colon T_x\zeta(v_x)=T_up_K^{LM} (V_u)\big\}\subset T_xM\times T_u(LM).
 \]
 Then we construct the following $K$-invariant distribution $H$ on $\zeta^*(LM)$: If $k\in K$ fullfils the condition $u=\widehat{\zeta}(x)\cdot k$, then
 \[
 H_{(x,u)}:=\big\{\big(v_x,T_{\widehat{\zeta}(x)}R_k\big(T_x\widehat{\zeta}(v_x)\big)\big)\colon v_x\in T_xM\big\}.
 \]
 It can be shown that it defines a flat connection on $\zeta^*(LM)$.
\end{proof}

So, in order to find a lift for the section $\zeta$, it is sufficient to construct a flat connection on the $K$-principal bundle $\zeta^*(LM)$.

To this end, we will define
\[
 \omega^\zeta:=\pi_{\kf}\circ ({\rm pr}_2 )^*\omega_0\in\Omega^1 (\zeta^*(LM),\kf )
\]
where $\omega_0\in\Omega^1 (LM,\mathfrak{gl}(m) )$ is a principal connection on $LM$ and $\pi_\kf\colon \mathfrak{gl}(m)\to\kf$ is the canonical projection onto~$\kf$. Lemma~\ref{lem:FlatConnectionIffSectionJ1pi} allows us to establish the following definition, inspired in the analogous concept from regular Routh reduction.

\begin{Definition}[flat condition for Palatini gravity]
 We will say that a metric $\zeta\colon M\to\Sigma$ satisfies the \emph{flat condition regarding the principal connection $\omega_0\in\Omega^1(LM,\mathfrak{gl}(m))$} if and only if the associated connection $\omega^\zeta$ is flat.
\end{Definition}

\begin{Remark}[flat condition and parallelizability]\label{rem:FlatConditionAndParallel}
 This condition yields to a relationship between the metric $\zeta\colon M\to\Sigma$ and the principal connection $\omega_0$; the physical relevance of this relationship remains unclear for the author. In order to get this condition mathematically, let us discuss the meaning of the bundle $\zeta^*(LM)$. By definition, $(x,u= (X_1,\dots,X_m))\in M\times LM$ belongs to $\zeta^*(LM)$ if and only $X_i\in T_xM$ for all $i=1,\dots,m$ and also
 \begin{gather}\label{eq:ZetaVsBasis}
 \zeta(x)^\sharp=\eta^{ij}X_i\otimes X_j.
 \end{gather}
 Here $g^\sharp$ indicates the contravariant $2$-tensor associated to the metric $g$; an equivalent way to express this is given by the formula
 \[
 \zeta(x)\circ (T_u\tau\otimes T_u\tau )=\eta_{ij} \theta^i\big|_u\otimes \theta^j\big|_u.
 \]
 It means in particular that we can identify
 \begin{gather*}
 \iota\colon \ \zeta^*(LM)\stackrel{\sim}{\longrightarrow} O_\zeta(M),
 \end{gather*}
 where $O_\zeta (M)\subset LM$ is the orthogonal subbundle associated to the metric $\zeta$. This identification goes as follows: Given $(x,u)\in\zeta^*(LM)$, equation~\eqref{eq:ZetaVsBasis} tells us that $u\in O_\zeta (M)$; conversely, given $u\in O_\zeta (M)$, the pair $ (\tau(u),u)$ belongs to $\zeta^*(LM)$. Therefore, we have the commutative diagram
 \[
 \begin{tikzcd}[row sep=1.3cm,column sep=1.5cm]
 \zeta^*(LM)
 \arrow[swap]{d}{\iota}
 \arrow{r}{{\rm pr}_2}
 &
 LM
 \\
 O_\zeta(M)
 \arrow[hook]{ur}{\text{inc}}
 &
 \end{tikzcd}
 \]
 Thus, under the previous identification, the $\mathfrak{gl}(m)$-valued $1$-form ${\rm pr}_2^*\omega_0$ is the pullback of the form $\omega_0$ to the subbundle $O_\zeta(M)$; because the decomposition
 \[
 \mathfrak{gl}(m)=\kf\oplus\pf
 \]
 is $K$-invariant, the $\kf$-valued $1$-form $\pi_{\kf}\circ({\rm pr}_2)^*\omega_0$ is a connection form on $O_\zeta(M)$. Thus, flatness of this connection is equivalent to $O_\zeta(M)$ being trivial. Now, it means that $M$ is parallelizable; therefore, we have proved that the flat condition for Palatini gravity is equivalent to parallelizability of the spacetime manifold~$M$. In particular, the reconstruction scheme can be carried out only locally.
\end{Remark}

Also, it is necessary to establish the following result regarding the map $F_\omega$.
\begin{Lemma}
 The following diagram commutes
 \begin{equation}\label{eq:F_OmegaCommutative_Diagram}
 \begin{tikzcd}[row sep=1.3cm,column sep=1.5cm]
 \cT_0
 \arrow[swap]{d}{p_K^{\cT_0}}
 \arrow{r}{F_\omega}
 &
 J^1\tau_\Sigma
 \arrow{d}{(\tau_\Sigma)_{10}}
 \\
 \cT_0/K
 \arrow{r}{\overline{\tau}'_{10}}
 &
 \Sigma
 \end{tikzcd}
 \end{equation}
\end{Lemma}
\begin{proof}
 In fact, for $j_x^1s\in\cT_0$ we have
 \begin{gather*}
 (\tau_\Sigma)_{10}\big(F_\omega\big(j_x^1s\big)\big)= (\tau_\Sigma)_{10}\big(j_x^1\big(p_K^{LM}\circ s\big)\big) =[s(x)]_K
 \end{gather*}
 and also
 \[
 \overline{\tau}'_{10}\big(p_K^{\cT_0}\big(j_x^1s\big)\big) =\overline{\tau}'_{10}\big(\big[j_x^1s\big]_K\big)= [s(x)]_K,
 \]
 and the lemma follows.
\end{proof}

With this in mind, we are ready to formulate the reconstruction side of this version of Routh reduction for Palatini gravity.

\begin{Theorem}[reconstruction in Palatini gravity]\label{thm:Reconstruction}
 Let $\zeta\colon M\to\Sigma$ be a metric satisfying the flat condition and the Einstein--Hilbert equations of motion. Then there exists a section
 \[
 \widehat{Z}\colon \ M\to\cT_0
 \]
 that is extremal of the Griffiths variational problem for Palatini gravity.
\end{Theorem}
\begin{proof} The holonomic lift
 \[
 j^1\zeta\colon \ M\to J^1\tau_\Sigma
 \]
 is extremal for the variational problem $\big(J^1\tau_\Sigma,\cL_{\rm EH}^{(1)},cI^\Sigma_{{\rm con}}\big)$; then, by Proposition~\ref{prop:FieldTheoryEqsEinsteinHilbertCase}, there exists a section
 \[
 \widetilde{\Gamma}\colon \ M\to W_{\rm EH}
 \]
 such that $\tau_{\rm EH}\circ\widetilde{\Gamma}=j^1\zeta$ and
 \begin{equation}\label{eq:WidetildeGammaLambdaEH}
 \widetilde{\Gamma}^* (X\lrcorner {\rm d}\lambda_{\rm EH} )=0
 \end{equation}
 for all $Z\in\mathfrak{X}^{V((\tau_\Sigma)_1\circ\tau_{\rm EH})}(W_{\rm EH})$.

 On the other hand, by Lemma \ref{lem:FlatConnectionIffSectionJ1pi} we have a lift
 \[
 \widehat{Z}\colon \ M\to\cT_0
 \]
 such that
 \[
 \overline{\tau}_{10}'\circ p_K^{\cT_0}\circ\widehat{Z}=\zeta;
 \]
 by diagram~\eqref{eq:F_OmegaCommutative_Diagram} we have that
 \begin{gather}\label{eq:HolonomicSectionZHat}
 \zeta=\overline{\tau}_{10}'\circ p_K^{\cT_0}\circ\widehat{Z}= (\tau_\Sigma )_{10}\circ F_\omega\circ\widehat{Z}.
 \end{gather}
 We will define the map
 \[
 \Gamma':=\big(\widehat{Z},\widetilde{\Gamma}\big)\colon \ M\to\cT_0\times W_{\rm EH}
 \]
 and show that it is a section of ${\rm pr}_1^\omega\colon F^*_\omega (W_{\rm EH})\to\cT_0$; namely, we have to show that
 \[
 F_\omega\circ\widehat{Z}=\pi_{\rm EH}\circ\widetilde{\Gamma}.
 \]
 It is important to this end to note that the conclusion of Proposition~\ref{prop:Metricity-Horizontal} can be translated to this context into
 \[
 Tp_K^{LM}\circ\omega_\pf=F_\omega^*\overline{\omega};
 \]
 moreover, by Lemma \ref{lem:FlatConnectionIffSectionJ1pi} we know that $\widehat{Z}^*\omega_\pf=0$, so
 \[
\big(F_\omega\circ\widehat{Z}\big)^*\overline{\omega}=\widehat{Z}^* \big(F_\omega^*\overline{\omega}\big)=Tp_K^{LM}\circ\big(\widehat{Z}^*\omega_\pf\big)=0.
 \]
 Then the section
 \[
 F_\omega\circ\widehat{Z}\colon \ M\to J^1\tau_\Sigma
 \]
 is holonomic; finally, from equation~\eqref{eq:HolonomicSectionZHat} we must conclude that
 \[
 F_\omega\circ\widehat{Z}=j^1\zeta.
 \]
 But $j^1\zeta=\pi_{\rm EH}\circ\widetilde{\Gamma}$ by construction of the section $\widetilde{\Gamma}$; then
 \[
 F_\omega\circ\widehat{Z}=\pi_{\rm EH}\circ\widetilde{\Gamma}
 \]
 and $\Gamma'$ is a section of $F_\omega^*\left(W_{\rm EH}\right)$, as required.

 Now define the section
 \[
 {\Gamma}:=\widetilde{F}_\omega\big(\widehat{Z},\widetilde{\Gamma}\big)\colon \ M\to W_{\rm PG};
 \]
 Therefore, for any $Z\in\mathfrak{X}^{V (\tau_1\circ\pi_{\rm PG})} (W_{\rm PG})$ that is $\big({\rm pr}_2^\omega\circ F_\omega^{-1}\big)$-projectable, we have that
 \begin{align*}
 \Gamma^* (Z\lrcorner {\rm d}\lambda_{\rm PG} )&= (\Gamma' )^*\big(\big(TF_\omega^{-1}\circ Z\big)\lrcorner {\rm d}F_\omega^*\lambda_{\rm PG}\big)
 =(\Gamma')^*\big(\big(TF_\omega^{-1}\circ Z\big)\lrcorner {\rm d}\big({\rm pr}_2^\omega\big)^*\lambda_{\rm EH}\big)\\
 &=\big({\rm pr}_2^\omega\circ\Gamma'\big)^*\big(\big(T{\rm pr}_2^\omega\circ TF_\omega^{-1}\circ Z\big)\lrcorner {\rm d}\lambda_{\rm EH}\big)\\
 &=\widetilde{\Gamma}^*\big(\big(T{\rm pr}_2^\omega\circ TF_\omega^{-1}\circ Z\big)\lrcorner {\rm d}\lambda_{\rm EH}\big) =0
 \end{align*}
 because $\widetilde{\Gamma}$ obeys equation~\eqref{eq:WidetildeGammaLambdaEH}.
\end{proof}

\section{Conclusions and outlook}\label{sec:conclusions-outlook}

We adapt the Routh reduction scheme developed in~\cite{Capriotti2019} to the case of affine gravity with vielbeins. It suggests that this formalism could be fit to deal with Griffiths variational problems more general than the classical, at least with cases when the differential restrictions are a subset of those imposed by the contact structure. Extensions of this scheme to gravity interacting with matter fields will be studied elsewhere.

\appendix
\section{An important algebraic result} First, we want to state the following algebraic proposition.

\begin{Proposition}\label{Prop:UniqueSolution}
Let $\left\{c_{ijk}\right\}$ be a set of real numbers such that
\[
\begin{cases}
c_{ijk}\mp c_{jik}=b_{ijk},\\
c_{ijk}\pm c_{ikj}=a_{ijk}
\end{cases}
\]
for some given set of real numbers $\{a_{ijk}\}$ and $\{b_{ijk}\}$ such that $b_{ijk}\mp b_{jik}=0$ and $a_{ijk}\pm a_{ikj}=0$. Then,
\[
c_{ijk}=\frac{1}{2} (a_{ijk}+a_{jki}-a_{kij}+b_{ijk}+b_{kij}-b_{jki} )
\]
is the unique solution for this linear system.
\end{Proposition}
\begin{proof}From first equation we see that
\[
\pm c_{jik}=c_{ijk}-b_{ijk}.
\]
The trick now is to form the following combination
\begin{gather*}
a_{ijk}+a_{jki}-a_{kij} =c_{ijk}\pm c_{ikj}+c_{jki}\pm c_{jik}-\left(c_{kij}\pm c_{kji}\right)
 =2c_{ijk}-b_{ijk}-b_{kij}+b_{jki},
\end{gather*}
where in the permutation of indices was used the remaining condition.
\end{proof}

\section{Proof of Proposition \ref{prop:FieldTheoryEqsWL}}\label{app:LiftToTorsionZero}

In order to do this proof, it will be necessary to bring some facts from \cite{doi:10.1142/S0219887818500445}. First, we have the bundle isomorphism on $\cT_0$
\[
 W_{\rm PG}\simeq E_2
\]
where $p_2'\colon E_2\to\cT_0$ is the vector bundle
\[
 E_2:=\wedge^{m-1}_1 (\cT_0 )\otimes S^*(m),
\]
with $S^*(m):=\big(\mathbb{R}^m\big)^*\odot\big(\mathbb{R}^m\big)^*$ the set of symmetric forms on $\mathbb{R}^m$, and
\[
 \wedge^{m-1}_1 (\cT_0 ):=\big\{\gamma\in\wedge^{m-1} (\cT_0)\colon \gamma \text{ is horizontal respect to the projection }\tau_1'\colon \cT_0\to M\big\}.
\]
The bundle $E_2$ is a bundle of forms with values in a vector space; therefore, it has a canonical $(m-1)$-form
\[
 \Theta:=\Theta_{ij}e^i\odot e^j.
\]
Using the structure equations for the canonical connection on $J^1\tau$ (pulled back to $\cT_0$), we have that the differential of the Lagrangian form $\lambda_{\rm PG}$ is given by
\begin{gather*}
 {\rm d}\lambda_{\rm PG}|_\rho=\big[2\eta^{kp}(\omega_\pf)_k^i\wedge\theta_{il}-(\omega_\pf)^s_s\wedge \eta^{kp}\theta_{kl}+\eta^{ip} \Theta_{il}|_{\beta}\big]\wedge\Omega^l_p \nonumber\\
\hphantom{{\rm d}\lambda_{\rm PG}|_\rho=}{} +\eta^{ik}\big[{\rm d}\Theta_{ij} |_{\beta}+\eta^{rq}\eta_{li} \Theta_{rj}|_{\beta}\wedge(\omega_\kf)^l_q -\Theta_{ip}|_{\beta}\wedge(\omega_\kf)^p_j\big] \wedge(\omega_\pf)^j_k.%\label{eq:FormulaFordLambda0}
\end{gather*}
The equations of motion
\begin{gather}\label{eq:LiftedEqsOfMotion}
\Gamma^*(X\lrcorner d\lambda_{\rm PG})=0,\qquad X\in\mathfrak{X}^{V(\tau_1'\circ\tau_{\rm PG})}(W_{\rm PG})
\end{gather}
are obtained by choosing a convenient set of vertical vector fields; because of the identification given above, it is sufficient to give a set of vertical vector fields on $\cT_0$ and on~$E_2$. It results that a~global basis of vertical vector fields on $\cT_0$ is
\[
 B':=\big\{A_{J^1\tau},M_{rs}\big(\theta^r,\big(E^s_j\big)_{LM}\big)^V\colon A\in\mathfrak{gl}(m),M_{pq}-M_{qp}=0\big\};
\]
in fact, the equation defining $\cT_0$
\[
 e_\rho^ke_{k\sigma}^\mu=e_\sigma^ke_{k\rho}^\mu
\]
is invariant by the ${\rm GL}(m)$-action, and also
\[
\big(\theta^r,\big(E^s_j\big)_{LM}^V\big)\cdot\big(e_\rho^ke_{k\sigma}^\mu-e_\sigma^ke_{k\rho}^\mu\big) =e^\mu_j\big(e^s_\sigma e^r_\rho-e^r_\sigma e^s_\rho\big).
\]
Given that $E_2$ is a vector bundle on $\cT_0$, any section $\beta\colon \cT_0\to E_2$ gives rise to a vertical vector field; the equations of motion associated to these kind of vector fields are the metricity conditions
\[
 \omega_\pf=0.
\]
Therefore, fixing an Ehresmann connection on the bundle $p_2'\colon E_2\to\cT_0$, we can produce the set of vertical vector fields on $E_2$
\[
(A_{\cT_0})^H,\qquad M:=M_{rs}\big(\theta^r,\big(E^s_j\big)_{LM}\big)^H;
\]
the equations of motion associated to $M$ are
\[
 \eta^{ks}\Gamma^*\big(M_{rs}\Theta_{kt}\wedge\theta^r\big)=0.
\]
The unique solution of these equations is $\Theta_{kt}=0$. In fact, by writing
\[
 \Gamma^*\Theta_{kt}=\eta^{lp}N_{ktp}\theta_l
\]
and taking into account the symmetry properties of $M_{pq}$, we have that the set of quantities $N_{pqr}$ must satisfy
\[
 N_{pqr}-N_{qpr}=0,\qquad N_{pqr}+N_{prq}=0;
\]
by Proposition~\ref{Prop:UniqueSolution}, it results that $N_{pqr}=0$, as desired. The rest of the equations of motion can be calculated is the same fashion that in the~$J^1\tau$ case; therefore, the equations~\eqref{eq:LiftedEqsOfMotion} are equivalent to the equations for the extremals of the Palatini variational problem.

\section{Proof of Proposition \ref{prop:Local_Horizontal_Lift}}\label{sec:proof-prop-ref}

First, let us write down
\[
 \left(\frac{\partial}{\partial x^\mu}\right)^H=M^\nu_\mu\frac{\partial}{\partial x^\nu}+N_{\mu k}^\nu\frac{\partial}{\partial e^\nu_k}.
\]
Then from
\begin{equation}\label{eq:HorizontalLiftProjected}
 Tp_K^{LM}\left(\left(\frac{\partial}{\partial x^\mu}\right)^H\right)=\frac{\partial}{\partial x^\mu}
\end{equation}
it results
\[
 M^\nu_\mu=\delta^\nu_\mu;
\]
the condition
\[
 \omega_K\left(\left(\frac{\partial}{\partial x^\mu}\right)^H\right)=0
\]
implies
\[
 N^\nu_{\mu k}\big(\eta^{pk}e_\nu^l-\eta^{lk}e_\nu^p\big)+\eta^{lq}e^p_\sigma f^\sigma_{q\mu}-\eta^{pq}e^l_\sigma f^\sigma_{q\mu}=0.
\]
In order to understand this equation for the unknowns~$N^\nu_{\mu k}$, let us change of variables through the formula
\[
 N^\nu_{\mu k}=g_{\sigma\rho}e^\rho_kN^{\nu\sigma}_\mu;
\]
in terms of these new variables, and the Christoffel symbols of the connection $\omega_0$
\[
 \overline{\Gamma}^\sigma_{\alpha\mu}=-e_\alpha^kf_{k\mu}^\sigma,
\]
the previous equations can be expressed as
\begin{align*}
 0&=g_{\alpha\rho}e^\rho_kN^{\nu\alpha}_\mu\big(\eta^{pk}e_\nu^l-\eta^{lk}e_\nu^p\big)-\eta^{lq}e^p_\sigma e_q^\alpha \overline{\Gamma}^\sigma_{\alpha\mu}+\eta^{pq}e^l_\sigma e^\alpha_q \overline{\Gamma}^\sigma_{\alpha\mu}\\
 &=N^{\nu\alpha}_\mu\big(g_{\alpha\rho}e^\rho_k\eta^{pk}e_\nu^l-g_{\alpha\rho}e^\rho_k\eta^{lk}e_\nu^p\big) +\big(\eta^{pq}e^\alpha_qe^l_\sigma-\eta^{lq}e^\alpha_qe^p_\sigma\big)\overline{\Gamma}^\sigma_{\alpha\mu}\\
 &=N^{\nu\alpha}_\mu\big(g_{\alpha\rho}g^{\rho\beta}e_\beta^pe_\nu^l-g_{\alpha\rho} g^{\rho\beta}e_\beta^le_\nu^p\big)+\big(e^p_\beta e^l_\sigma-e^l_\beta e^p_\sigma\big)g^{\alpha\beta}\overline{\Gamma}^\sigma_{\alpha\mu}\\
 &=N^{\sigma\beta}_\mu\big(e_\beta^pe_\sigma^l-e_\beta^le_\sigma^p\big)+\big(e^p_\beta e^l_\sigma-e^l_\beta e^p_\sigma\big)g^{\alpha\beta}\overline{\Gamma}^\sigma_{\alpha\mu}\\
 &=\big(e_\beta^pe_\sigma^l-e_\beta^le_\sigma^p\big) \big(N^{\sigma\beta}_\mu+g^{\alpha\beta}\overline{\Gamma}^\sigma_{\alpha\mu}\big).
\end{align*}
The operator in the left is essentially an antisymmetrizator because of the formula
\[
 e^\mu_pe^\nu_l\big(e_\beta^pe_\sigma^l-e_\beta^le_\sigma^p\big)=\delta_\beta^\mu \delta_\sigma^\nu-\delta_\beta^\nu \delta_\sigma^\mu;
\]
therefore
\begin{gather}\label{eq:NFromGamma}
 N^{\sigma\beta}_\mu+g^{\alpha\beta}\overline{\Gamma}^\sigma_{\alpha\mu}=S^{\sigma\beta}_\mu,
\end{gather}
where
\[
 S^{\sigma\beta}_\mu-S^{\beta\sigma}_\mu=0.
\]
Finally, from the condition \eqref{eq:HorizontalLiftProjected} we obtain
\[
 N_{\mu k}^\sigma\big(\eta^{kq}e_q^\rho\delta_\sigma^\alpha+\eta^{kq}e_q^\alpha\delta_\sigma^\rho\big)=0
\]
or, in terms of the variables $N^{\nu\sigma}_\mu$
\begin{align*}
 0&=g_{\nu\beta}e^\beta_kN_{\mu}^{\sigma\nu}\big(\eta^{kq}e_q^\rho\delta_\sigma^\alpha +\eta^{kq}e_q^\alpha\delta_\sigma^\rho\big)
 =N_{\mu}^{\sigma\nu}\big(g_{\nu\beta}e^\beta_k\eta^{kq}e_q^\rho\delta_\sigma^\alpha +g_{\nu\beta}e^\beta_k\eta^{kq}e_q^\alpha\delta_\sigma^\rho\big)\\
 &=N_{\mu}^{\sigma\nu}\big(g_{\nu\beta}g^{\beta\rho}\delta_\sigma^\alpha +g_{\nu\beta}g^{\beta\alpha}\delta_\sigma^\rho\big)
 =N_{\mu}^{\sigma\nu}\big(\delta_{\nu}^{\rho}\delta_\sigma^\alpha+\delta_{\nu}^{\alpha}\delta_\sigma^\rho\big).
\end{align*}
From equation \eqref{eq:NFromGamma} it results that
\[
 g^{\alpha\beta}\overline{\Gamma}^\sigma_{\alpha\mu} +g^{\alpha\sigma}\overline{\Gamma}^\beta_{\alpha\mu}-2S^{\sigma\beta}_\mu=0,
\]
or equivalently
\[
 N^{\sigma\beta}_\mu=\frac{1}{2}\big(g^{\alpha\sigma}\overline{\Gamma}^\beta_{\alpha\mu} -g^{\alpha\beta}\overline{\Gamma}^\sigma_{\alpha\mu}\big).
\]
Therefore, we have
\[
 \left(\frac{\partial}{\partial x^\mu}\right)^H=\frac{\partial}{\partial x^\mu}+\frac{1}{2}g_{\beta\rho}e^\rho_k\big(g^{\alpha\sigma}\overline{\Gamma}^\beta_{\alpha\mu} -g^{\alpha\beta}\overline{\Gamma}^\sigma_{\alpha\mu}\big)\frac{\partial}{\partial e^\sigma_k}.
\]

Additionally, we need to construct the horizontal lifts
\[
 \left(\frac{\partial}{\partial g^{\mu\nu}}\right)^H=P_{\mu\nu}^\sigma\frac{\partial}{\partial x^\sigma}+Q_{\mu\nu k}^\sigma\frac{\partial}{\partial e^\sigma_k}
\]
with $P_{\mu\nu}^{\sigma}-P_{\nu\mu}^{\sigma}=0$, $Q_{\mu\nu k}^{\sigma}-Q_{\nu\mu k}^\sigma=0$. The equation
\begin{gather*} %\label{eq:HorizontalLiftE}
 Tp_K^{LM}\left(\left(\frac{\partial}{\partial g^{\mu\nu}}\right)^H\right)=\frac{\partial}{\partial g^{\mu\nu}}
\end{gather*}
and the identity \eqref{eq:E_Projection} imply
\[
 P_{\mu\nu}^{\sigma}\frac{\partial}{\partial x^\sigma}+Q_{\mu\nu k}^{\sigma}\left(\eta^{kq}e_q^\rho\delta_\sigma^\alpha+\eta^{kq}e_q^\alpha\delta_\sigma^\rho\right)\frac{\partial}{\partial g^{\alpha\rho}}=\frac{\partial}{\partial g^{\mu\nu}},
\]
namely
\[
 P_{\mu\nu}^{\sigma}=0
\]
and (given the symmetry properties of $g^{\mu\nu}$)
\begin{gather}
 \frac{1}{2}\big(\delta_\mu^\alpha\delta^\rho_\nu+\delta_\nu^\alpha\delta^\rho_\mu\big) =Q_{\mu\nu k}^{\sigma}\big(\eta^{kq}e_q^\rho\delta_\sigma^\alpha+\eta^{kq}e_q^\alpha\delta_\sigma^\rho\big)
 =\eta^{kq}e_q^\rho Q_{\mu\nu k}^\alpha+\eta^{kq}e_q^\alpha Q_{\mu\nu k}^{\rho}.\label{eq:Condition_One_E}
\end{gather}
The horizontality condition
\[
 \omega_K\left(\left(\frac{\partial}{\partial g^{\mu\nu}}\right)^H\right)=0
\]
will be equivalent to
\begin{gather}
 \big(\eta^{pk}e^l_\sigma-\eta^{lk}e^p_\sigma\big)Q_{\mu\nu k}^\sigma=0.
 \label{eq:Condition_Two_E}
\end{gather}
These conditions can be understood by introducing the variables
\[
 Q_{\mu\nu}^{\alpha\rho}:=\eta^{kq}e_q^\rho Q_{\mu\nu k}^{\alpha};
\]
then, equation \eqref{eq:Condition_One_E} becomes
\[
 Q_{\mu\nu}^{\alpha\rho}+Q_{\mu\nu}^{\rho\alpha}=\frac{1}{2}\left(\delta_\mu^\alpha\delta^\rho_\nu+\delta_\nu^\alpha\delta^\rho_\mu\right)
\]
and equation \eqref{eq:Condition_Two_E} is equivalent to
\[
 Q_{\mu\nu}^{\alpha\rho}-Q_{\mu\nu}^{\rho\alpha}=0.
\]
Therefore
\begin{gather*}
 Q_{\mu\nu}^{\alpha\rho}+Q_{\mu\nu}^{\rho\alpha} =\frac{1}{2}\big(Q_{\mu\nu}^{\alpha\rho}+Q_{\mu\nu}^{\rho\alpha}\big) +\frac{1}{2}\big(Q_{\mu\nu}^{\alpha\rho}-Q_{\mu\nu}^{\rho\alpha}\big)
 =\frac{1}{4}\big(\delta_\mu^\alpha\delta^\rho_\nu+\delta_\nu^\alpha\delta^\rho_\mu\big)
\end{gather*}
and so
\begin{gather*}
 \left(\frac{\partial}{\partial g^{\mu\nu}}\right)^H =Q_{\mu\nu k}^{\alpha}\frac{\partial}{\partial e^\alpha_k}
 =\frac{1}{4}\eta_{kl}e^l_\rho\big(\delta_\mu^\alpha\delta^\rho_\nu +\delta_\nu^\alpha\delta^\rho_\mu\big)\frac{\partial}{\partial e^\alpha_k}
 =\frac{1}{4}g_{\rho\beta}e^\beta_k\big(\delta_\mu^\alpha\delta^\rho_\nu +\delta_\nu^\alpha\delta^\rho_\mu\big)\frac{\partial}{\partial e^\alpha_k}.
\end{gather*}

\subsection*{Acknowledgements}

The author thanks the CONICET and UNS for financial support, and Eduardo Garc\'{\i}a-Tora\~no for valuable discussion regarding aspects of Routh reduction contained in this article, as well as for pointing me out reference \cite{kharlamov_characteristic_1977}. Also, the author would like to warmly thank the referees for the care they put in reviewing this work. The article has been greatly improved by their suggestions.

\pdfbookmark[1]{References}{ref}
\LastPageEnding

\end{document}